\documentclass{article}

\usepackage{fullpage}
\usepackage{amsmath}
\usepackage{amsthm}
\usepackage{amssymb}
\usepackage{amsfonts}
\usepackage{graphics}
\usepackage{color}
\usepackage{graphicx}

\usepackage{latexsym}
\usepackage{euscript}
\usepackage{amstext}
\usepackage{color}
\usepackage{dsfont}
\usepackage{xspace}
\usepackage{url,hyperref}

\usepackage{algorithm}
\usepackage{algpseudocode}

\newtheorem{lemma}{Lemma}[section]
\newtheorem{theorem}[lemma]{Theorem}

\newtheorem{definition}[lemma]{Definition}
\newtheorem{proposition}[lemma]{Proposition}
\newtheorem{claim}[lemma]{Claim}

\newcommand{\etal}{et al.\ }
\newcommand{\eps}{\epsilon}
\newcommand{\ignore}[1]{}

\newcommand{\cS}{\mathcal{S}}

\newcommand{\cA}{\mathcal{A}}
\newcommand{\cG}{\mathcal{G}}

\newcommand{\Ex}{\mathbb{E}}
\newcommand{\cQ}{\mathcal{Q}}

\newcommand{\true}{\texttt{True}}
\newcommand{\false}{\texttt{False}}

\newcommand{\initOneLiners}{%
    \setlength{\itemsep}{0pt}
    \setlength{\parsep }{0pt}
    \setlength{\topsep }{0pt}
}

\newenvironment{OneLiners}[1][\ensuremath{\bullet}]
    {\begin{list}
        {#1}
        {\initOneLiners}}
    {\end{list}}

\newcommand{\cov}{\texttt{cov}}
\newcommand{\alg}{\textsf{ALG}}
\newcommand{\opt}{\textsf{OPT}\xspace}

\newcommand{\gmssc}{\textsf{GMSSC}\xspace}
\newcommand{\mssc}{\textsf{MSSC}\xspace}
\newcommand{\sop}{\textsf{SOP}\xspace}
\newcommand{\lscp}{\textsf{MLSC}\xspace}
\newcommand{\mlsc}{\textsf{MLSC}\xspace}

\newcommand{\lcst}{\textsf{LCST}\xspace}
\newcommand{\lgst}{\textsf{LGST}\xspace}
\newcommand{\cst}{\textsf{CST}\xspace}
\newcommand{\gst}{\textsf{GST}\xspace}

\newcommand{\sr}{\textsf{SR}\xspace}
\newcommand{\wssr}{\textsf{WSSR}\xspace}

\newcommand{\alglscp}{\textsf{ALG-MLSC}}
\newcommand{\cp}{\textsf{ALG-SOP}\xspace}
\newcommand{\ag}{\textsf{ALG-AG}\xspace}
\newcommand{\alglcst}{\textsf{ALG-LCST}}

\newcommand{\algkrs}{\textsf{ALG-KRS}}
\newcommand{\agsto}{\textsf{ALG-AG-STO}\xspace}

\newcommand{\lplcst}{\mathsf{LP}_\mathsf{LCST}}
\newcommand{\lpcst}{\mathsf{LP}_\mathsf{CST}}

\def\sse{\subseteq}

\title{Minimum Latency Submodular Cover\thanks{A preliminary version appeared in the proceedings of ICALP 2012.}}
\author{
Sungjin Im \thanks{Department of Computer Science, Duke University, USA.  \texttt{sungjin@cs.duke.edu}. This work was partially supported by NSF grant CCF-1016684.} \and 
Viswanath Nagarajan \thanks{IBM T. J. Watson Research Center, USA. \texttt{viswanath@us.ibm.com}}  
\and 
Ruben van der Zwaan
\thanks{Maastricht University, The Netherlands \texttt{r.vanderzwaan@maastrichtuniversity.nl}}
}

\date{\today}

\begin{document}

\maketitle

\begin{abstract}
We study the Minimum Latency Submodular Cover problem (\lscp), which consists of a metric $(V,d)$ with source $r\in V$ and $m$ monotone submodular
functions $f_1, f_2, ..., f_m: 2^V \rightarrow [0,1]$. The goal is to find a path originating at $r$ that minimizes the
total cover time of all functions. This generalizes well-studied problems, such as Submodular Ranking~\cite{AzarG11} and Group Steiner Tree~\cite{GKR00}. We give a polynomial time $O(\log \frac{1}{\eps} \cdot \log^{2+\delta}
|V|)$-approximation algorithm for \lscp, where $\epsilon>0$ is the smallest non-zero marginal increase of any $\{f_i\}_{i=1}^m$ and $\delta>0$ is any constant.

We also consider the Latency Covering Steiner Tree problem ($\lcst$), which is the special case of \mlsc where the $f_i$s are multi-coverage functions. 
This is a common generalization of the Latency Group Steiner Tree~\cite{GuptaNR10a,ChakrabartyS11} and Generalized Min-sum Set Cover~\cite{AzarGY09,BansalGK10} problems. We obtain an $O(\log^2|V|)$-approximation algorithm for \lcst. 

Finally we study a natural stochastic extension of the Submodular Ranking problem, and obtain an adaptive algorithm with an $O(\log 1/ \eps)$ approximation ratio, which is best possible. This result
also generalizes some previously studied stochastic optimization problems, such as Stochastic Set Cover~\cite{GoemansV06} and
Shared Filter Evaluation~\cite{MunagalaSW07,LiuPRY08}.
\end{abstract}

\section{Introduction}
Ordering a set of elements so as to be simultaneously good for several valuations is an important issue in web-search
ranking and broadcast scheduling. 
A formal model for this was introduced by Azar, Gamzu and Yin~\cite{AzarGY09} where they
studied the Multiple Intents Re-ranking problem (a.k.a. Generalized Min Sum Set Cover~\cite{BansalGK10}). In this problem, a set of elements is to be displayed to $m$ different users, each of whom wants to see some threshold number of elements from its subset of interest. The objective is to compute an ordering that minimizes the average (or total) overhead of the users, where the overhead corresponds to the position in the ordering when the user is satisfied.

Subsequently, Azar and Gamzu~\cite{AzarG11} studied a generalization, the Submodular Ranking problem, where the interests of users are represented by arbitrary (monotone) submodular functions. Again, the objective is to order the elements so as to minimize the total overhead, where now the overhead of a user is the position when its utility function is ``covered''. An interesting feature of this problem is that it generalizes both the minimum set cover~\cite{J74} and min-sum set cover~\cite{BBHST98,FeigeLT04} problems.

In this paper, we extend both of these models to the setting of metric switching costs between elements. This allows us to
handle additional issues such as:
\begin{OneLiners}
\item  {\em Data locality:} it takes $d(i,j)$ time to read or transmit data $j$ after data $i$.
\item {\em Context switching:} it takes $d(i,j)$ time for a user to parse data $j$ when scheduled after data $i$.
\end{OneLiners}

From a theoretical point of view, these problems generalize a number of previously studied problems and our results unify/extend techniques used in different settings.

We introduce and study the Minimum Latency Submodular Cover problem (\lscp), which is the metric version of Submodular
Ranking~\cite{AzarG11}, and its interesting special case, the Latency Covering Steiner Tree problem (\lcst),
which extends Generalized Min-Sum Set Cover~\cite{AzarGY09,BansalGK10}. The formal definitions follow shortly, in the next subsection. We obtain poly-logarithmic approximation guarantees for both problems. We remark that due to a relation to the well-known
Group Steiner Tree~\cite{GKR00} problem, any significant improvement on our results would lead to a similar
improvement for Group Steiner Tree. The \lscp problem is a common generalization of several previously studied
problems~\cite{GKR00,KonjevodRS02,FeigeLT04,GuptaNR10a,ChakrabartyS11,AzarGY09,AzarG11}; see also
Figure~\ref{fig:problems}.

In a somewhat different direction, we also study the Weighted Stochastic Submodular Ranking problem, where elements are stochastic and the goal is to
adaptively schedule elements so as to minimize the expected total cover time. We obtain an $O(\log
\frac1\epsilon)$-approximation algorithm for this problem, which is known to be best possible even in the deterministic
setting~\cite{AzarG11}. This result
also generalizes many previously studied stochastic optimization problems~\cite{GoemansV06,MunagalaSW07,LiuPRY08}.

\subsection{Problem Definitions} 
	\label{sec:problem-definitions}

We now give formal definitions of the problems considered in this paper. The problems followed by $^*$ are those for which we obtain the first non-trivial results. Several other problems are also discussed since those are important special cases of our main problems. The relationships between these problems are shown pictorially  in Figure~\ref{fig:problems}.

A function $f: 2^{V} \rightarrow
\mathbb{R}_+$ is {\em submodular} if, for any $A, B \subseteq V$, $f(A) + f(B) \geq f(A \cup B) + f(A \cap B)$; and it
is {\em monotone} if for any $A \subseteq B$, $f(A) \leq f(B)$. We assume some familiarity with submodular
functions~\cite{schrijver}. 

\begin{figure}[h]
    \begin{center}
      \includegraphics[width=0.85\textwidth]{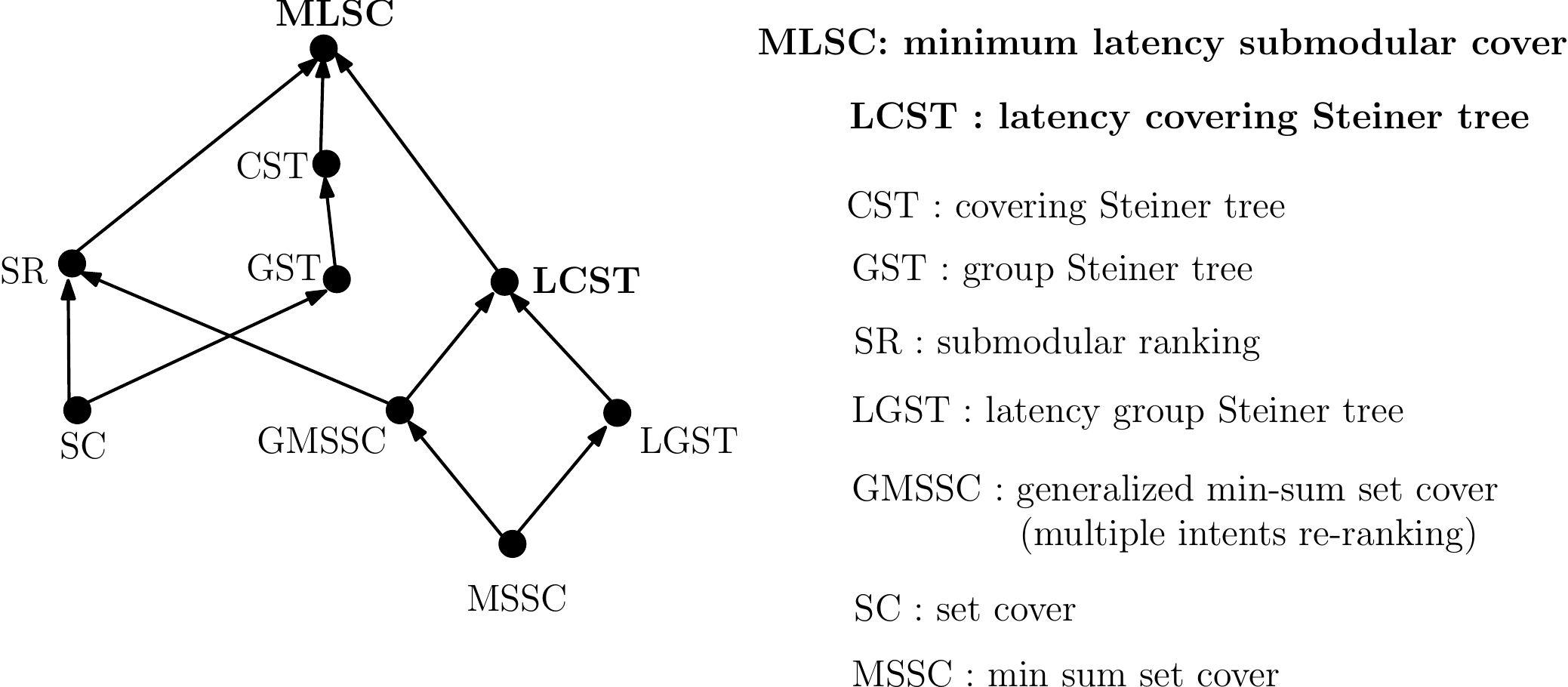}
    \end{center}
    \caption{\label{fig:problems} An arrow from $X$ to $Y$ means $X$ is a special case of $Y$.}
\end{figure}

\medskip
\noindent
\textbf{Minimum Latency Submodular Cover$^*$ (\mlsc):} There is a ground set $V$ of elements/vertices and $d:{V \choose 2} \rightarrow \mathbb{R}_+$ is a distance
function. We assume that $d$ is symmetric and satisfies the triangle inequality. In addition there is a specified root
vertex $r\in V$. There are $m$ monotone submodular functions $f_1,\ldots,f_m : 2^V\rightarrow \mathbb{R}_+$
representing the valuations of different users. We assume, without loss of generality by truncation, that $f_i(V)=1$ for all $i\in[m]$.\footnote{Throughout the paper, for any integer $\ell\ge 1$, we denote $[\ell]:=\{1,2,\ldots,\ell\}$.}
Function $f_i$ is said to be {\em covered} (or satisfied) by set $S\sse V$ if $f_i(S)=1=f_i(V)$. The {\em cover time} of function
$f_i$ in a path $\pi$ is the length of the shortest prefix of $\pi$ that has $f_i$ value one, i.e.
$$\min \,\, \bigg\{ t \,:\, f_i\left(\{v\in V : v \mbox{ appears within distance $t$ on } \pi\}\right)=1 \bigg\}.$$
The objective in the Minimum Latency Submodular Cover problem is to compute a path originating at $r$ that
minimizes the sum of cover times of all functions. A technical parameter that we use to measure performance (which also appears in~\cite{AzarG11,Wolsey82}) is $\epsilon$ which is defined to be the smallest non-zero marginal increase of any function $\{f_i\}_{i=1}^m$. 


\medskip
\noindent
\textbf{Generalized Min-Sum Set Cover (\gmssc):} Given a ground set $V$ and $m$ subsets $\{g_i\sse V\}_{i=1}^m$ with respective requirements $\{k_i\}_{i=1}^m$, the goal is to find a linear ordering of $V$ that minimizes the sum of cover times. A subset $g_i$ is said to be covered when at least $k_i$ elements from $g_i$ have appeared.  Min-Sum Set Cover (\mssc) is the special case when $\max_{i} k_i=1$.

\medskip
\noindent
\textbf{Submodular Ranking (\sr):} Given a ground set $V$ and $m$ monotone submodular functions $f_1,\ldots,f_m : 2^V\rightarrow \mathbb{R}_+$, the goal is compute a linear ordering of $V$ that minimizes the sum of cover times of all functions. The cover time of a function is the minimum number of elements in a prefix that has function value at least one. This is a special case of \mlsc when metric $d$ is uniform. The set cover problem is a special case of \sr when there is a single submodular function (which is also a coverage function). \gmssc is another special case of \sr, where each subset $g_i$ corresponds to the submodular function $f_i(S)=\min\{|g_i\cap S|/k_i,1\}$.

\medskip
\noindent
\textbf{Group Steiner Tree (\gst):} Given a metric $(V,d)$ with root $r\in V$ and $N$ groups of vertices $\{g_i\sse V\}_{i=1}^N$, the goal is to find a minimum length tree
containing $r$ and at least one vertex from each of the $N$ groups. Observe that an $r$-rooted tree can be converted into a path starting from $r$ with at most a factor two loss in the total length, and vice versa. Thus \gst is a special case of \mlsc when there is only a single submodular function 
$$ f_1(S) = \frac1N \, \sum_{i=1}^N \min\{|g_i\cap S|,\,1\}.$$
Note that $f_1(S')=1$ if and only if $S'\bigcap g_i$ is nonempty for all $i\in[N]$. 

\medskip
\noindent
\textbf{Covering Steiner Tree (\cst):} This is a generalization of \gst with the same input as above, where each group $g_i$ is also associated with a requirement $k_i$. The goal here is to find a minimum length tree that contains $r$ and at least $k_i$ vertices from group $g_i$, for all $i\in[N]$. We recover \cst as a special case of \mlsc by setting 
$$f_1(S) = \frac1N \, \sum_{i=1}^N \min\left\{\frac{|g_i\cap S|}{k_i},\,1\right\}.$$ Note that now $f_1(S')=1$ if and only if $|S'\bigcap g_i| \geq k_i$ for all $i\in[N]$.

\medskip
\noindent
\textbf{Latency Group Steiner Tree (\lgst):} This is a variant of the group Steiner tree problem. Given a metric $(V,d)$ with root $r$ and $N$ groups of vertices $\{g_i\sse V\}_{i=1}^N$, the goal is to find a path $\pi$ originating from $r$ that minimizes the sum of cover times of the groups. (A group $g_i$ is covered at the shortest prefix of  $\pi$ that contains at least one vertex from $g_i$.) Note that \mssc is the special case when the metric is uniform.

\medskip
\noindent
\textbf{Latency Covering Steiner Tree$^*$ (\lcst):} The input to this problem is the same as for \lgst with additional requirements $\{k_i\}_{i=1}^N$ corresponding to each group. The objective is again a path $\pi$ originating from $r$ that minimizes the sum of cover times, where group $g_i$ is covered at the shortest prefix of  $\pi$ that contains at least $k_i$ vertices from $g_i$. Clearly, \lgst is the special case of \lcst where all requirements $k_i = 1$. \gmssc is also a special case when the metric is uniform. We obtain \lcst as a special case of \mlsc with $m=N$ functions and $f_i(S)=\min\{|g_i\cap S|/k_i,\,1\}$ for all $i\in[N]$.


\medskip
\noindent
\textbf{Weighted Stochastic Submodular Ranking$^*$ (\wssr):} This is a stochastic generalization of the submodular ranking problem. We are given a set $V$ of stochastic elements (random variables), each having an independent distribution over a certain domain $\Delta$. The submodular functions are also defined on the ground set $\Delta$, i.e. $f_1, ..., f_m: 2^\Delta
\rightarrow [0,1]$. In addition, each element $i\in V$ has a deterministic time $\ell_i$ to be
scheduled. The realization (from $\Delta$) of any
element is known immediately after scheduling it. The goal is to find an adaptive ordering of $V$ that
minimizes the total expected cover time. Since elements are stochastic, it is possible that a function is never
covered: in such cases we just fix the cover time to be $\sum_{i\in V} \ell_i$ (which is the total duration of any
schedule).

We will be concerned with {\em adaptive} algorithms. Such an algorithm is allowed to decide the next element to schedule
based on the instantiations of the previously scheduled elements. This models the setting where the algorithm can
benefit from user feedback.

\wssr generalizes the Stochastic Set Cover studied in \cite{GoemansV06}. Interestingly,
it also captures some variants of Stochastic Set Cover that have applications in processing multiple queries with
probabilistic information \cite{MunagalaSW07, LiuPRY08}. Various applications of \wssr are discussed in more detail in Section~\ref{sec:stochastic}.

\subsection{Our Results and Techniques}

Our first result is on the Minimum Latency Submodular Cover problem (\lscp) problem.

\begin{theorem}\label{thm:lscp1}
For any constant $\delta>0$, there is an $O(\log \frac{1}{\eps} \cdot \log^{2+\delta} |V|)$-approximation algorithm for
the Minimum Latency Submodular Cover problem.
\end{theorem}

Note that in the special case of Group Steiner Tree, this result is larger only by a factor of $O(\log^{\delta}
|V|)$ than its best known approximation ratio of $O(\log N\,\log^2|V|)$, due to Garg, Konjevod and Ravi~\cite{GKR00}. Our algorithm
uses the framework of~\cite{AzarG11} and the Submodular Orienteering problem (\sop)~\cite{ChekuriP05} as a
sub-routine. The input to \sop consists of metric $(V,d)$, root $r$, monotone submodular function $f:2^V\rightarrow
\mathbb{R}_+$ and length bound $B$. The goal is to find a path originating at $r$ having length at most $B$ that
maximizes $f(S)$, where $S\sse V$ is the set of vertices visited in the path. Specifically, we show that a
$(\rho,\sigma)$-bicriteria approximation algorithm\footnote{Given any instance of \sop, such an algorithm returns a path of length at most $\sigma\cdot B$ and function value at least $\opt/\rho$.} for \sop can be used to obtain an $O(\rho\,\sigma\cdot
\log\frac1\epsilon)$-approximation algorithm for \lscp. To obtain Theorem~\ref{thm:lscp1} we use an $( O(1), \,
O(\log^{2+\delta} |V|) )$-bicriteria approximation for \sop that follows from~\cite{CZ05,CEK06}.

Our algorithm for \lscp is an extension of the elegant ``adaptive residual updates scheme'' of Azar and
Gamzu~\cite{AzarG11} for Submodular Ranking (i.e. uniform metric \lscp). As shown in~\cite{AzarG11}, an
interesting aspect of this problem is that the natural greedy algorithm, based on absolute contribution of elements,
performs very poorly. Instead they used a modified greedy algorithm that selects one element at a time according to
residual coverage. In the \lscp setting of general metrics, our algorithm uses a similar residual coverage {\em
function} to repeatedly augment the solution. However our augmentations are paths of geometrically increasing lengths,
instead of just one element. A crucial point in our algorithm is that the residual coverage functions  are always
submodular, and hence we can use Submodular Orienteering (\sop) in the augmentation step.

We remark that the approach of covering the maximum number of functions within geometrically increasing lengths fails because the residual
coverage function here is non-submodular; in fact as noted in~\cite{BansalGK10} this subproblem contains the difficult
dense-$k$-subgraph problem (even for the special case of Generalized Min-Sum Set Cover with requirement two). We also note that the choice
of our (submodular) residual coverage function ultimately draws on the submodular ranking algorithm~\cite{AzarG11}.

The analysis in~\cite{AzarG11} was based on viewing the optimal and approximate solutions as histograms. This approach
was first used in this line of work by Feige, Lov{'a}sz and Tetali~\cite{FeigeLT04} for the Min-Sum Set Cover problem (see
also~\cite{BBHST98}). This was also the main framework of analysis in~\cite{AzarGY09} for Generalized Min-Sum Set
Cover and then for Submodular Ranking~\cite{AzarG11}. However, these proofs have been increasingly difficult as
the problem in consideration adds more generality. Instead we follow a different and more direct approach that is
similar to the analysis of Minimum Latency Travelling Salesman problem, see eg.~\cite{CGRT03,FHR07}. In fact, the proof of Theorem~\ref{thm:lscp1} is enabled by a new simpler analysis of the Submodular Ranking algorithm~\cite{AzarG11}.

\medskip
Our second result is a better approximation ratio for the Latency Covering Steiner tree (\lcst) problem. Note that
Theorem~\ref{thm:lscp1} implies directly an $O(\log k_{max}\cdot \log^{2+\delta}|V|)$-approximation algorithm for
\lcst, where $k_{max}=\max_{i=1}^N k_i$. 
\begin{theorem}\label{thm:lcst}
There is an $O(\log^2|V|)$-approximation algorithm for Latency Covering Steiner Tree.
\end{theorem}
The main idea in this result is a new LP relaxation for Covering Steiner Tree (using {\em Knapsack Cover} type
inequalities~\cite{CFLP00}) having a poly-logarithmic integrality gap. This new LP might also be of some independent interest. The previous
algorithms~\cite{KonjevodRS02,GuptaS06} for covering Steiner tree were based on iteratively solving an LP with large
integrality gap.  However, this approach does not seem suitable to the {\em latency} version we consider. Our new LP relaxation for Covering Steiner Tree (\cst) is crucial for obtaining the approximation stated in Theorem~\ref{thm:lcst}. 
As shown in~\cite{N09}, any improvement over Theorem~\ref{thm:lcst} even in the $k_{max}=1$ special case
(i.e. Latency Group Steiner Tree) would yield an improved approximation ratio for Group Steiner Tree, which
is a long-standing open question. 

\medskip
Our final result is for the Weighted Stochastic Submodular Ranking problem. As shown in~\cite{GoemansV06,GolovinK10}, even
special cases of this problem have polynomially large adaptivity gap (ratio between the optimal non-adaptive and
adaptive solutions)\footnote{A non-adaptive solution is just a fixed linear ordering of the elements, whereas an adaptive solution can select the next element based on previous instantiations.}. This motivates adaptive algorithms, and we obtain the following result in Section~\ref{sec:stochastic}.
\begin{theorem}
    \label{thm:sto-sr}
There is an adaptive $O(\log \frac{1}{\eps})$-approximation algorithm for the Weighted Stochastic Submodular Ranking problem.
\end{theorem}
In particular, we show that the natural stochastic extension of the algorithm from~\cite{AzarG11} achieves this
approximation factor. We remark that the analysis in~\cite{AzarG11} of deterministic submodular ranking required unit
costs, whereas Theorem~\ref{thm:sto-sr} holds for the stochastic setting even with non-uniform costs $\{\ell_i\}$.

As mentioned before, our results generalize the results in \cite{GoemansV06,MunagalaSW07,LiuPRY08} which study (some
variants of) Stochastic Set Cover. Our analysis is arguably simpler and more transparent than \cite{LiuPRY08}, which
gave the first tight analysis of these problems. We note that~\cite{LiuPRY08} used an intricate charging scheme with
``dual prices'' and it does not seem directly applicable to general submodular functions.

\medskip
We note that our techniques do not extend directly to the stochastic \mlsc problem (on general metrics), and obtaining a poly-logarithmic approximation here seems to require additional ideas.

\subsection{Previous Work}
	\label{sec:previous}

The first poly-logarithmic approximation for Group Steiner Tree was $O(\log N\, \log^2|V|)$, obtained by Garg et al.~\cite{GKR00}. This is still the best known bound. Chekuri, Even and Kortsarz~\cite{CEK06} gave a combinatorial 
algorithm that achieved a slightly weaker approximation ratio (the algorithm in~\cite{GKR00} was LP-based). This
combinatorial approach was extended in Calinescu and Zelikovsky~\cite{CZ05} to the problem of covering any submodular
function in a metric space. We use this algorithm in the submodular orienteering (\sop) subroutine for our \lscp result. For
\sop an $O(\log |V|)$-approximation is known due to Chekuri and Pal~\cite{ChekuriP05}, but with a {\em quasi-polynomial} running time.  We note that an $\Omega(\log^{2-\delta}|V|)$ hardness of approximation is known
for Group Steiner Tree (even on tree metrics) due to Halperin and Krauthgamer~\cite{HK03}.

The Covering Steiner Tree problem was introduced by Konjevod, Ravi and Srinivasan~\cite{KonjevodRS02}, which can be viewed as the
multicover version of Group Steiner Tree. They gave an $O(\log(Nk_{max})\,\log^2|V|)$-approximation using an
LP-relaxation. However the LP used in~\cite{KonjevodRS02} has a large $\Omega(k_{max})$ integrality gap; they got
around this issue by iteratively solving a suitable sequence of LPs. They also extended the randomized rounding
analysis from~\cite{GKR00} to this context. Later, Gupta and Srinivasan~\cite{GuptaS06} improved the approximation
bound to $O(\log N \,\log^2|V|)$, removing the dependence on the covering requirements. This algorithm was also based
on solving a similar sequence of LPs; the improvement was due to a combination of threshold rounding and randomized
rounding. In this paper, we give a stronger LP relaxation for Covering Steiner Tree based on so-called Knapsack-Covering-inequalities (abbreviated to KC-inequalities), that has an
$O(\log N \,\log^2|V|)$ integrality gap. 

The Stochastic Set Cover problem (which is a special case of Weighted Stochastic Submodular Ranking) was introduced by Goemans and Vondr{\'a}k~\cite{GoemansV06}. Here each set covers a random subset of items, and the goal is to minimize the expected cost of a set cover. \cite{GoemansV06} showed a large adaptivity gap for Stochastic Set Cover, and gave a logarithmic
approximation for a relaxed version where each stochastic set can be added multiple times. A related problem in context of fast query evaluation was studied in~\cite{MunagalaSW07}, where the authors gave a triple logarithmic approximation. This bound was improved
to the best-possible logarithmic ratio by Liu, Parthasarathy, Ranganathan and Yang~\cite{LiuPRY08}; this result was also applicable to stochastic set cover (where each set can be added at most once).
Another related paper is by Golovin and Krause~\cite{GolovinK10}, where they defined a general property ``adaptive
submodularity'' and showed nearly optimal approximation guarantees for several objectives (max coverage, min-cost cover
and min sum cover). The most relevant result in~\cite{GolovinK10} to \wssr is the 4-approximation
for Stochastic Min Sum Set Cover.
This approach required a {\em fixed} submodular function $f$ such that the objective is to minimize $\Ex\left[ \sum_{t\ge 0}
f(\overline{V}) - f(\overline{\pi}_t)\right]$ where $\overline{\pi}_t$ is the realization of elements scheduled
within time $t$ and $\overline{V}$ denotes the realization of  all elements. However, this is
not the case even for the special case of Generalized Min-Sum Set Cover with requirements two. Recently, Guillory and Bilmes~\cite{GB11} studied the Submodular Ranking problem in an online regret setting, which is different from the adaptive model we consider.


\subsection{Organization}
In Section~\ref{sec:submodularranking} we revisit the Submodular Ranking problem and give an easier and perhaps more intuitive analysis of the algorithm from~\cite{AzarG11}. This simpler analysis is then used in the algorithms for 
Minimum Latency Submodular Cover (Theorem~\ref{thm:lscp1}) and Weighted Stochastic Submodular Ranking (Theorem~\ref{thm:sto-sr}), that appear in Sections \ref{sec:submodularcover} and \ref{sec:stochastic} respectively. 
Section~\ref{sec:latencycovering} contains the improved approximation algorithm for Latency Covering Steiner Tree (Theorem~\ref{thm:lcst}) which makes use 
of a new linear programming relaxation for Covering Steiner Tree. The section on \lcst can be read independently of the other three sections.

\section{Simpler Analysis of the Submodular Ranking Algorithm} \label{sec:submodularranking}

In this section, we revisit the Submodular Ranking problem~\cite{AzarG11}. Recall that the input consists of a ground
set  $V:=[n]$ of elements and monotone  submodular functions $f_1, f_2, ... f_m: 2^{[n]} \rightarrow [0,1]$ with
$f_i(V)=1,\,\forall i\in[m]$. The goal is to find a complete linear ordering of the elements that minimizes the total
cover time of all functions. The cover time $\cov(f_i)$ of $f_i$ is defined as the smallest index  $t$ such that the
function $f_i$ has value 1 for the first $t$ elements in the ordering. We also say that an element $e$ is scheduled at
time $t$ if it is the $t$-th element in the ordering. It is assumed that each function $f_i$ satisfies the following
property: for any $S \supseteq S'$, if $f_i(S) - f_i(S') >0$ then it must be the case that $f_i(S) - f_i(S') \geq
\eps$, where $\eps >0$ is a constant that is uniform for all functions $f_i$. This is a useful parameter in describing
the performance guarantee.

Azar and Gamzu~\cite{AzarG11} gave a modified greedy-style algorithm with an approximation factor of $O(\log
\frac{1}{\eps})$ for Submodular Ranking.
Their analysis was histogram-based and fairly involved. In this section, we give an alternate shorter proof of their result. Our analysis also extends to the more general \lscp
problem which we study in the next section. The algorithm \ag from~\cite{AzarG11} is given below. In the output,
$\pi(t)$ denotes the element that appears in the $t$-th time slot.

\begin{algorithm}[h!] \caption{$\ag$} \label{alg:rsv} 
    \textbf{INPUT}: Ground set $[n]$; monotone submodular functions $f_i : 2^{[n]} \rightarrow [0,1],  i \in [m]$
    \begin{algorithmic}[1]
    \State $S \leftarrow \emptyset$
    \For {$t = 1$ to $n$}
	    \State Let $f^S(e) := \sum_{i \in [m], f_i(S) <1} \,\, \frac{f_i(S \cup \{e\}) - f_i (S)}{1 - f_i(S)}$
	    \State $e = \arg\max_{e \in [n] \setminus S} \,\, f^S(e)$
	    \State $S \leftarrow S\bigcup \{e\}$
	    \State $\pi(t) \gets e$
    \EndFor
    \end{algorithmic}
    \textbf{OUTPUT}: A linear ordering $\langle \pi(1),\pi(2),\ldots,\pi(n)\rangle$ of $[n]$.
\end{algorithm}

\begin{theorem}[\cite{AzarG11}]
    \label{thm:ag-main}
$\ag$ is an $O(\ln (\frac{1}{\eps}))$-approximation algorithm for Submodular Ranking.
\end{theorem}

Let $\alpha: = 1 + \ln( \frac{1}{\eps})$. To simplify notation, without loss of generality, we assume that $\alpha$ is an integer.
Let $R(t)$ denote the set of functions that are \emph{not satisfied} by \ag earlier than time $t$; $R(t)$ includes the
functions that are satisfied exactly at time $t$. For notational convenience, we use $i \in R(t)$ interchangeably with
$f_i \in R(t)$. Analogously, $R^*(t)$ is the set of functions that are not satisfied in the optimal solution before time $t$. Note that algorithm's objective $\alg = \sum_{t \geq 1} |R(t)|$
and the optimal value $\opt = \sum_{t \geq 1} |R^*(t)|$. We will be interested in the number of unsatisfied functions at times $\{8
\alpha 2^j \, :\, j \in \mathbb{Z}_+\}$ by $\ag$ and the number of unsatisfied functions at times $\{2^j\, :\, j \in
\mathbb{Z}_+\}$ by the optimal solution. Let $R_j := R(8 \alpha 2^j)$ and $R^*_j = R^*(2^j)$ for all integer $j\ge 0$. It is important to note
that $R_j$ and $R^*_j$ are concerned with different times. For notational simplicity, we let $R_{-1} := \emptyset$.

We show the following key lemma. Roughly speaking, it says that the number of unsatisfied functions by $\ag$
diminishes quickly unless it is comparable to the number of unsatisfied functions in $\opt$.

\begin{lemma}    \label{lem:ag-main}
    For any $j \geq 0$, we have $|R_j| \leq \frac{1}{4} |R_{j-1}| + |R^*_j|$.
\end{lemma}
\begin{proof}
When $j=0$ the lemma holds trivially. Now consider any integer $j \geq 1$ and time step $t \in [8\alpha 2^{j-1}, 8 \alpha 2^{j})$.  Let $S_{t-1}$ denote the set of elements that $\ag$ schedules before time $t$ and let $e_t$ denote the element that $\ag$ schedules exactly at time $t$. Let $E_j$  denote the set of elements that $\ag$ schedules until time $8 \alpha 2^j$. Let $E^*_j$ denote the set of elements that $\opt$ schedules until time $2^j$. Recall that $\ag$ picks $e_t$ as an element $e$ that maximizes
$$f^{S_{t-1}}(e) := \sum_{i \in [m]: f_i(S_{t-1}) <1} \frac{f_i(S_{t-1} \cup \{e\}) - f_i (S_{t-1})}{1 - f_i(S_{t-1})}$$

This leads us to the following proposition.

\begin{proposition} \label{prop:sub-a}
For any $j \geq 1$, time step $t \in [8\alpha 2^{j-1}, 8 \alpha 2^{j})$ and $e \in E^*_j$, we
have $f^{S_{t-1}}(e_t) \geq f^{S_{t-1}}(e)$.
\end{proposition}
\begin{proof}
Since $\ag$ has chosen to schedule element $e_t$ over all elements $e \in E^*_{j}  \setminus S_{t-1}$, we know that the
claimed inequality holds for any $e \in E^*_{j} \setminus S_{t-1}$.  Further, the inequality holds for any element  $e$
in $S_{t-1}$, since $f^{S_{t-1}}(e) = 0$ for such an element $e$.
\end{proof}

By taking an average over all elements in $E^*_{j}$, we derive
\begin{align}
 f^{S_{t-1}}(e_t) \geq & \ \frac{1}{|E^*_{j}|} \sum_{e \in E^*_{j} } f^{S_{t-1}} (e) \nonumber \\
\geq & \ \frac{1}{|E^*_{j}|} \sum_{e \in E^*_{j} } \,\, \sum_{i \in R_j \setminus R^*_j} \frac{f_i(S_{t-1} \cup \{e\}) -
f_i (S_{t-1})}{1 - f_i(S_{t-1})} \label{eqn:ag-core}
\end{align}

Observe that in~\eqref{eqn:ag-core}, the inner summation only involves functions $f_i$ for which  $f_i(S_{t-1}) <1$. This is because for any $i \in  R_j$, function $f_i$ is not covered before time $8 \alpha 2^j$ and $t < 8 \alpha 2^j$. Due to submodularity of each function $f_i$, we have that

\begin{equation*}
(\ref{eqn:ag-core}) \quad \geq\quad \frac{1}{|E^*_{j}|}  \sum_{i \in R_j \setminus R^*_j} \frac{f_i(S_{t-1} \cup E^*_{j} ) - f_i (S_{t-1})}{1 - f_i(S_{t-1})} \quad 
=\quad  \frac{1}{|E^*_{j}|} \sum_{i \in R_j \setminus R^*_j}  1 \quad \geq \quad \frac{|R_j| - |R^*_{j}|}{|E^*_{j}|}
\end{equation*}

The equality is due to the fact that for any $i \notin R^*_j$, $f_i(E^*_{j}) =1$ and each function $f_i$ is monotone.
Hence:
\begin{equation}
\sum_{8 \alpha  \cdot 2^{j-1} \leq t < 8 \alpha \cdot  2^{j}} f^{S_{t-1}}(e_t)  \quad \geq  \quad \frac{8 \alpha (2^{j} -2^{j-1})}{|E^*_{j}|} (|R_j| - |R^*_{j}|) 
\quad =  \quad 4 \alpha (  |R_j| - |R^*_{j}|),\label{eqn:ag-lb}
\end{equation}
where we used $|E^*_j|=2^j$. We now upper bound the left-hand-side of (\ref{eqn:ag-lb}). To this end, we need the
following claim from~\cite{AzarG11}.

\begin{claim}[Claim 2.3 in~\cite{AzarG11}] \label{lem:log-ub} Given a monotone function $f:2^{[n]} \rightarrow [0,1]$ with $f([n])=1$
and sets $\emptyset = S_0 \sse S_1\sse \cdots \sse S_\ell\sse [n]$, we have (using the convention $0/0=0$)
$$\sum_{k =1}^\ell \frac{ f( S_k) - f(S_{k-1})}{1 - f(S_{k-1})} \,\, \leq \,\, 1 + \ln \frac{1}{\delta}.$$ Here $\delta >0$ is
such that for any $A \subseteq B$, if $ f(B) - f(A) > 0$ then $f(B) - f(A) \geq \delta$.
\end{claim}
\begin{proof}
We give a proof for completeness. We can assume, without loss of generality, that $S_\ell=[n]$. Order the values in the set $\{ f(S_k) \; | \; 0  \leq k \leq \ell \} \setminus
\{1\}$ in increasing order
to obtain $\beta_0 < \beta_1 < \ldots  < \beta_H$. By the assumption, we have $\beta_0 \ge 0$ and $\beta_H \leq 1-\delta$ (moreover, $\beta_h - \beta_{h-1} \geq \delta, \,\forall h \in [H]$). We will show that 
    $$\sum_{h = 1}^{H} \frac{ \beta_h - \beta_{h-1} } { 1 - \beta_{h-1}} \quad \leq  \quad \ln \frac{1}{\delta}$$
Since $f(S_\ell)=1$, the summation we want to bound has an additional term of $\frac{1-\beta_H}{1-\beta_H}=1$.

Knowing that the function $u(x)=\frac{1}{1-x}$ is increasing for $x \in [0, 1)$,  we derive
\begin{align*}
\sum_{h = 1}^{H} \frac{ \beta_h - \beta_{h-1} } { 1 - \beta_{h-1}} \,\, &= \,\, \sum_{h = 1}^{H} \int_{x = \beta_{h-1}}^{\beta_h} \frac{1}{ 1 - \beta_{h-1}} \ \texttt{d}x 
\,\,  \leq \,\, \sum_{h = 1}^{H} \int_{x = \beta_{h-1}}^{\beta_h} \frac{1}{1 - x} \ \texttt{d}x \,\,  
 = \,\, \int_{x = 0}^{\beta_H} \frac{1}{1 - x} \ \texttt{d}x \\
&= \ln \left(\frac{1- \beta_0}{1- \beta_H} \right) \quad \leq \quad \ln\frac{1}{\delta}
\end{align*}
This proves the claim.\end{proof}

Note that any function $f_i$ not in $R_{j-1}$ does not contribute to the left-hand-side of (\ref{eqn:ag-lb}) since any such function $f_i$ was already covered before time $8\alpha\, 2^{j-1}\le t$. Further, knowing  by
Claim~\ref{lem:log-ub} that each function $f_i \in R_{j-1}$ can add at most $\alpha := 1+\ln \frac{1}{\eps}$, we can
upper bound the left-hand-side of (\ref{eqn:ag-lb}) by $\alpha |R_{j-1}|$. Formally,
\begin{align}
\sum_{8 \alpha \cdot 2^{j-1} < t \leq 8 \alpha  \cdot 2^{j}} f^{S_{t-1}}(e_t) & = \sum_{8 \alpha \cdot 2^{j-1} < t \leq 8 \alpha  \cdot 2^{j}} \,\, \sum_{i \in R_{j-1}:  f_i(S_{t-1}) <1}  \frac{f_i(S_{t-1} \cup \{e_t\}) - f_i
(S_{t-1})}{1 -f_i(S_{t-1})} \notag \\
& \leq \sum_{i \in R_{j-1}} \,\, \sum_{t \geq 1: f_i(S_{t-1}) <1 }  \frac{f_i(S_{t-1} \cup \{e_t\}) - f_i (S_{t-1})}{1 - f_i(S_{t-1})} \notag \\
& \leq  \,\,\alpha |R_{j-1}|     \label{eqn:ag-ub}
\end{align}

From (\ref{eqn:ag-lb}) and (\ref{eqn:ag-ub}) we obtain $4 \alpha (  |R_j| - |R^*_j|) \leq  \alpha |R_{j-1}| $ which completes the proof of Lemma~\ref{lem:ag-main}. \end{proof}

\noindent Now we can prove Theorem~\ref{thm:ag-main} using Lemma~\ref{lem:ag-main}.
\begin{proof}[Proof of Theorem~\ref{thm:ag-main}.]
{\small \begin{eqnarray*}    \alg
    &=&  \sum_{j \geq 0} \quad \sum_{8\alpha  2^j \leq t < 8 \alpha  2^{j+1}} |R(t)| \quad + \quad \sum_{1 \leq t < 8 \alpha} |R(t)|  \nonumber \\
    &\leq& \sum_{j \geq 0} \,\, 8 \alpha (2^{j+1} - 2^j)  |R_j| \quad  + \quad    8\alpha \opt \;\;\;\;\;\;   \mbox{[Since $|R(t)|$ is non-increasing, and $|R(1)| \leq m \leq \opt$]} \nonumber\\
    &=& 8 \alpha \sum_{j \geq 0} \, 2^{j+1} \left( |R_j| -\frac{1}{4} |R_{j-1}| \right)  \quad + \quad 8\alpha \opt   \\
    &\leq& 8 \alpha  \sum_{j \geq 0} \, 2^{j+1}  |R^*_j|  \quad +\quad 8 \alpha \opt  \quad\quad \mbox{[By Lemma~\ref{lem:ag-main}]} \\
    &\leq&  8 \alpha  \sum_{j \geq 1} 4 \quad \sum_{2^{j-1} \leq t < 2^j}  |R^*(t)|  \,\,+ \,\, 16 \alpha |R^*_0|  \quad + \quad 8\alpha \opt  \quad\quad \mbox{[Since $|R^*(t)|$ is non-increasing]} \\
    &\leq&  32\alpha \opt \,\, +\,\,  24 \alpha \opt.\\
\end{eqnarray*}}
Thus we obtain $\alg\le 56\alpha\,\opt$, which proves Theorem~\ref{thm:ag-main}.
\end{proof}

\section{Minimum Latency Submodular Cover} \label{sec:submodularcover}
Recall that in the Minimum Latency Submodular Cover problem (\lscp), we are given a metric $(V,d)$ with root $r\in V$ and $m$
monotone submodular functions $f_1, f_2, ..., f_m: 2^{V} \rightarrow [0, 1]$. 
Without loss of generality, by scaling, we assume that all distances $d(\cdot,\cdot)$ are integers. The objective in \mlsc is to
find a path starting at $r$ that minimizes the total cover time of all functions.

As mentioned earlier, our algorithm for $\lscp$ uses as a subroutine an algorithm for the Submodular Orienteering
problem $(\sop)$. In this problem, given metric $(V,d)$, root $r$, monotone submodular function $g:2^V\rightarrow \mathbb{R}_+$ and
bound $B$, the goal is to compute a path $P$ originating at $r$ that has length at most $B$ and maximizes $g(V(P))$
where $V(P)$ is the set of vertices covered by $P$. We assume a $(\rho,\,\sigma)$-bicriteria approximation
algorithm $\cp$ for \sop. That is, on any \sop instance, \cp returns a path $P$ of length at most $\sigma\cdot B$ and $g(V(P))\ge\opt/\rho$, where $\opt$ is the optimal value obtained by any length $B$ path. We recall the following known results on \sop.
\begin{theorem}[\cite{CZ05}] \label{thm:CZ05} For any constant $\delta>0$ there is a polynomial time  $(O(1),$ $O(\log^{2+\delta}|V|))$ bicriteria
approximation algorithm for the Submodular Orienteering problem.
\end{theorem}

\begin{theorem}[\cite{ChekuriP05}] \label{thm:ChekuriP05} There is a quasi-polynomial time  $O(\log |V|)$
approximation algorithm for the Submodular Orienteering problem.
\end{theorem}

We now describe our algorithm $\alglscp$  for $\lscp$ that uses the $(\rho,\sigma)$ bicriteria approximation algorithm
$\cp$. Here $\alpha = 1 + \ln \frac{1}{\eps}$. Note the difference from the submodular ranking
algorithm~\cite{AzarG11}: here each augmentation is a path possibly covering several vertices. Despite the similarity
of $\alglscp$ to the min-latency TSP type algorithms~\cite{CGRT03,FHR07} an important difference is that we {\em do
not} try to directly maximize the number of covered functions in each augmentation: as noted before this subproblem is
at least as hard as dense-$k$-subgraph, for which the best approximation ratio known is only polynomial~\cite{BCCFV10}.
Instead we maximize in each step some proxy residual coverage function $f^S$ that suffices to eventually cover all
functions quickly. This function is a natural extension of the single-element coverage values used in
\ag~\cite{AzarG11}. It is important to note that in Line (4), $f^S(\cdot)$ is defined adaptively  based on the current
set $S$ of visited vertices in each iteration. Moreover, since each function $f_i$ is monotone and submodular, so is $f^S$ for
any $S\sse V$. In Step~\ref{alg-mlsc:3}, $\pi \cdot P$ denotes the concatenation of paths $\pi$ and $P$.

\begin{algorithm}[h!] \caption{\alglscp} \label{alg:lsop}
\textbf{INPUT}: $(V, d), r \in V ; \{f_i: 2^V \rightarrow [0,1]\}_{i =1}^m$.
    \begin{algorithmic}[1]

    \State $S \leftarrow \emptyset$, $\pi\gets \emptyset$.

    \For {$k = 0, 1, 2, ...$} \label{alg-mlsc:phase}

	\For{$u=1,2,\ldots, 4 \alpha \rho$}    

\State \label{alg-mlsc:1}
Define submodular function $$f^S(T) :=  \sum_{i \in [m], f_i(S) <1} \frac{f_i(S \cup
T) - f_i (S)}{1 - f_i(S)}, \quad \mbox{ for all }T\sse V.$$

\State \label{alg-mlsc:2}
Use $\cp$ to find a path $P$ of length at most $\sigma\cdot 2^{k}$ starting from
$r$ that $\rho$-approximately maximizes $f^S(V(P))$ where $V(P)$ is the set of nodes visited by $P$.

\State \label{alg-mlsc:3}
$S \leftarrow S \cup V(P)$ and $\pi\gets \pi \cdot P$.

    \EndFor
    \EndFor

    \end{algorithmic}
    \textbf{OUTPUT}: Output solution $\pi$.
\end{algorithm}

We prove the following theorem, which implies Theorem~\ref{thm:lscp1}.
\begin{theorem}
    \label{thm:lsop}
    $\alglscp$ is an $O( \alpha \rho\sigma)$-approximation algorithm for Minimum Latency Submodular Cover.
\end{theorem}

We now analyze \alglscp.  We say that the algorithm is in the $j$-th phase, when the variable $k$ of the for loop
in Step~\ref{alg-mlsc:phase} has value $j$. Note that there are $4\alpha\rho$ iterations of Steps~\ref{alg-mlsc:1}-\ref{alg-mlsc:2} in each phase.
\begin{proposition}
    \label{prop:mlsc-1}
Any vertex $v$ added to $S$ in the $j$-th phase is visited by $\pi$ within $16 \alpha \rho \sigma \cdot 2^j$.
\end{proposition}
\begin{proof}
The final solution is a concatenation of the paths that
were found in Step~\ref{alg-mlsc:3}. Since all these paths are stitched at the root $r$, the length of $\pi$ at the end of phase $j$
is at most $\sum_{k=1}^j 2\cdot 4\alpha\rho\cdot \sigma 2^k \le 16 \alpha \rho \sigma \cdot 2^j$. 
\end{proof}

Let $R(t)$ denote the set of  (indices of) the functions that are not covered by $\alglscp$ earlier than time  $t$;
$R(t)$ includes the functions that are covered exactly at time $t$ as well. We interchangeably use  $i \in R(t)$
and $f_i \in R(t)$. Let $R_j :=R_j( 16 \alpha \rho \sigma\, 2^j)$. Similarly, we let
$R^*(t)$ denote the set of functions that are not covered by $\opt$ earlier than  time $t$ and let $R^*_j = R^*( 2^j)$. Let $R_{-1} := \emptyset$.

We show the following key lemma. It shows that the number of uncovered functions by $\alglscp$ must decrease fast
as $j$ grows, unless the corresponding number in the optimal solution is comparable.

\begin{lemma}
    \label{lem:lsop-main}
    For any $j \geq 0$, we have $|R_j| \leq \frac{1}{4} |R_{j-1}| + |R^*_{j}|$.
\end{lemma}
\begin{proof}
The lemma trivially holds when $j =0$. Now consider any fixed phase $j \geq 1$. Let $S_0$ denote the set of vertices
that were added to $S$ up to the end of phase $j-1$. Let $H = 4 \alpha \rho$ and $T_1, T_2, ..., T_{H}$ be the sets of
vertices that were added in Line (6) in the $j$-th phase. Let $S_h = S_0 \cup T_1 \cup T_2 \cup ... \cup T_h$, $\forall
1 \leq h \leq H$. We prove Lemma~\ref{lem:lsop-main} by lower and upper bounding the quantity
$$\Delta_j \,\, := \,\, \sum_{h = 1}^{H} f^{S_{h-1}}(T_h) \,\, = \,\, \sum_{h = 1}^{H} \,\,\, \sum_{i \in [m]: f_i(S_{h-1}) <1} \frac{f_i (S_h) -
f_i(S_{h-1})}{1 - f_i(S_{h-1})},$$
which is intuitively the total amount of ``residual requirement''  that is covered by the algorithm in phase $j$.

We first lower bound $\Delta_j$. Let $T^*$ denote the set of vertices that $\opt$ visited within time $2^j$. Observe
that in Line (5), $\alglscp$ could have visited all nodes in $T^*$ by choosing $P$ as $\opt$'s prefix of length $2^j$. Via the approximation guarantee of $\cp$, we obtain

\begin{proposition}
    For any $h \in [H]$, we have $f^{S_{h-1}}(T_h)  \geq \frac{1}{\rho} \cdot f^{S_{h-1}}(T^*)$.
\end{proposition}


We restrict our concern to the functions in $R_j \setminus R^*_j$. Observe that for any $i \in R_j$ and $h\in[H]$, $f_i
(S_{h-1}) <1$ and that for any $i \notin R^*_j$, $f_i (T^*) = 1$. Hence by summing the inequality in the above
proposition over all functions $f_i$ in $R_j \setminus R^*_j$, we have
\begin{eqnarray}
\Delta_j & \geq &\frac{1}{\rho} \sum_{h=1}^{H} f^{S_{h-1}}(T^*) 
\quad \geq  \quad \frac{1}{\rho} \sum_{h=1}^{H} \,\, \sum_{i \in R_j \setminus R^*_j} \frac{ f_i(T^* \cup S_{h-1}) - f_i(S_{h-1})}{1 -   f_i(S_{h-1})}  
\quad \geq  \quad \frac{1}{\rho} \sum_{h=1}^{H} \,\, \sum_{i \in R_j \setminus R^*_j} 1 \notag\\ 
& \geq &\frac{H}{\rho} (|R_j| - |R^*_j|) \quad =  \quad 4 \alpha   (|R_j| - |R^*_j|) \label{eqn:lsop-ub}
\end{eqnarray}

We now upper bound $\Delta_j$. Note that for any $i \notin  R_{j-1}$, $f_i(S_0) =1$ and therefore $f_i$
does not contribute to $\Delta_j$. For any $i \in R_{j-1}$, the total contribution of $f_i$ to $\Delta_j$ is at most
$\alpha$ by Claim~\ref{lem:log-ub}. Hence,
\begin{equation}
    \label{eqn:lsop-lb}
    \Delta_j \leq \alpha |R_{j-1}|
\end{equation}

Combining (\ref{eqn:lsop-ub}) and (\ref{eqn:lsop-lb}) completes the proof of Lemma~\ref{lem:lsop-main}.
\end{proof}

Finally, we can use Lemma~\ref{lem:lsop-main} to prove Theorem~\ref{thm:lsop} exactly as we proved Theorem~\ref{thm:ag-main} in
the previous section using Lemma~\ref{lem:ag-main}. We omit repeating the calculations here.

\ignore{
\begin{proof}[Proof of Theorem~\ref{thm:lsop}]
\begin{eqnarray*}    && \alg \\
    &=&  \sum_{j \geq 0} \quad \sum_{16\alpha \rho \sigma 2^j \leq t < 16 \alpha\rho \sigma  2^{j+1}} |R(t)| \quad + \quad \sum_{0 \leq t < 16 \alpha \rho \sigma} |R(t)|  \nonumber \\
    &\leq& \sum_{j \geq 0} \,\, 16 \alpha \rho \sigma (2^{j+1} - 2^j)  |R_j| \quad  + \quad    16 \alpha \rho \sigma \opt \;\;\;\;\;\;  \\
    && \mbox{[$|R(t)|$ is non-increasing, and $\forall t \geq 0$, $|R(t)| \leq m \leq \opt$]} \nonumber\\
    &=& 16 \alpha \rho \sigma \sum_{j \geq 0} \, 2^{j+1} \left( |R_j| -\frac{1}{4} |R_{j-1}| \right)  \quad + \quad 16 \alpha \rho \sigma \opt   \\
    &\leq& 16 \alpha \rho \sigma \sum_{j \geq 0} \, 2^{j+1}  |R^*_j|  \quad +\quad 16 \alpha \rho \sigma \opt  \quad \mbox{[Lemma~\ref{lem:lsop-main}]} \\
    &\leq&  16 \alpha \rho \sigma  \sum_{j \geq 1} 4  \left( \sum_{2^{j-1} \leq t < 2^j}  |R^*(t)| \right) + 32 \alpha \rho \sigma |R^*_0|   + 16 \alpha \rho \sigma \opt  \\
    &\leq&  56 \alpha \rho \sigma \opt \,\, +\,\,  48 \alpha \rho \sigma \opt \quad \mbox{[$|R^*(t)|$ is non-increasing]}
\end{eqnarray*}
\end{proof}
}

\section{Latency Covering Steiner Tree} \label{sec:latencycovering}
In this section, we give consider the Latency Covering Steiner Tree problem (\lcst), which is an
interesting special case of $\mlsc$.   Recall that the input to \lcst consists  of a symmetric metric $(V, d)$, root $r
\in V$ and a collection $\cG$ of groups, where each group $g\in \cG$ is a subset of vertices  with an associated requirement $k_g$.
The goal is  find a path staring from $r$ that minimizes the total cover time of all groups.  We say that group $g$ is
covered
at the earliest time $t$ when the path within distance $t$ visits at least $k_g$ vertices in $g$. We give an $O(\log
g_{max} \cdot \log |V|)$-approximation algorithm for this problem where $g_{max} := \max_{g \in \cG} |g|$ is the
maximum group size. This would prove Theorem~\ref{thm:lcst}.

\paragraph{Simplifying assumptions.}
Following \cite{KonjevodRS02, GuptaS06}, without loss of generality, we assume that:
\begin{enumerate}
    \item The metric is induced by a tree $T = (V, E)$ with root $r$ and weight $w_e$ on each edge $e \in E$.
    \item Every vertex in a group is a leaf, i.e. has degree one in $T$.
    \item The groups in $\cG$ are disjoint.
    \item Every vertex of degree one lies in some group.
\end{enumerate}
The only non-trivial assumption is the first one, which uses tree embedding~\cite{FakcharoenpholRT04} to reduce general metrics to trees, at the loss of an $O(\log |V|)$ approximation factor. In the rest of this section, we work with such instances of \lcst and obtain an $O(\log g_{max})$-approximation algorithm.

\ignore{For assumption $(1)$, a  $\rho$-approxi\-mation on tree instances immediately implies a randomized $O(\rho \log |V|)$-approximation on general instances. Hence we focus on giving an $O(\log m \cdot  \log g_{max})$-approximation on tree instances. Assumption $(2)$ follows from the following operation. Suppose there is a vertex $v$ that is in at least one group but with degree at least two. Create an extra vertex $v'$ and insert the edge $\{v,v'\}$ with weight $0$ and finally replace $v$ by $v'$ in all groups $v$ was present.  Assumption $(3)$ we have a degree $1$ vertex $v$ that is in groups $g_1,...,g_k$. Add a vertex $v_i$ and edge $\{v,v_i\}$ with weight $0$ and add $v_i$ to group $g_i$ for  $1\leq i \leq k$, finally remove $v$ from all groups. Assumption $(3)$ is due to the fact that if a vertex of degree $1$ is not in any group, the incident edge is not needed in the optimal solution.}

We first discuss a new LP relaxation for the covering Steiner tree problem in Subsection~\ref{subsec:lcst-LP0}, which can be shown to have a poly-logarithmic integrality gap. Next, in Subsection~\ref{subsec:lcst-LP} extend this idea to obtain an LP relaxation for latency covering Steiner tree. In Subsection~\ref{subsec:lcst-algo} we present our rounding algorithm for \lcst, and finally Subsection~\ref{subsec:lcst-analysis} contains the analysis of the algorithm.
\subsection{New LP Relaxation for \cst}\label{subsec:lcst-LP0}
Recall that the input to Covering Steiner Tree 
consists  of a metric $(V, d)$ with root $r$ and a collection of groups $\cG \subseteq 2^V$ where each group $g \in \cG$ is associated with a requirement $k_g$. The goal is to  find a minimum cost $r$-rooted tree that includes $r$ and at least $k_g$ vertices from each group $g$. Although an $O( \log m \cdot \log g_{max} \cdot \log n)$-approximation is known for \cst~\cite{GuptaS06}, there was no (single) linear program known to have a poly-logarithmic integrality gap. Previous results on \cst relied on an LP with large $\Omega(k_{max})$ integrality gap~\cite{KonjevodRS02}. 

We introduce stronger constraints, that yield an LP for \cst with integrality gap $O( \log m \cdot \log g_{max} \cdot
\log n)$. This new LP is an important ingredient in our algorithm for \lcst, and might also be useful in
other contexts.


\newcommand{\cut}{\texttt{cut}}
Let $L$ denote the set of leaves in $V$. Because of the above simplifying assumptions, we can label each  vertex $v$
in a group with a unique leaf-edge incident on it, and vice versa. We abuse notation by allowing $j \in L$ to denote
both the leaf-vertex and its unique incident edge.   For any subset of leaves
$L'\subseteq L$, let $\cut(r, L')$ denote the family of all edge-subsets whose removal separates  the root $r$ from all vertices in
$L'$.


We formulate the following linear programming relaxation for \cst on tree instances.
\begin{align}\tag{$\mathsf{LP}_\mathsf{CST}$} \label{lp:cst}
  \min \quad 				&  \sum_{e \in E} w_e x_e &\\
  \text{s.t.} \quad &  x_{pe(e)} \geq  x_e   &  \forall  e \in E \label{LP-cst-3}  \\
					          & (k_g - |A|) \sum_{j \in B  \setminus L} x_j +  \sum_{j \in B \cap (L \setminus A)} x_j  \geq  k_g - |A|
& \forall g \in \cG, \forall A \subseteq g, \forall  B \in \cut(r,g \setminus A) & \label{LP-cst-2}  \\
& x_e  \in \,\, [ 0, 1] &\forall e \in E \nonumber
\end{align}

\paragraph{Validity of ~\ref{lp:cst}.} We first argue that this is a valid relaxation. Consider any instance of $\cst$ on trees and a fixed feasible solution
(tree) $\tau^*$, which gives a natural integral solution: $x_e = 1$ if and only if $e \in \tau^*$.
We focus on constraints (\ref{LP-cst-2}), since the other constraints are obviously satisfied. Consider any $g \in \cG$,
$A \subseteq g$ and $B \in \cut(r, g \setminus A)$. Let $\tau^*(E \setminus A)$ denote the subtree induced by the edges in both $\tau^*$ and  $(E \setminus A)$, i.e. $\tau^* \bigcap (E\setminus A)$.
Note that $\tau^*(E \setminus A)$ is connected, since $A$ consists only of leaf edges. Since $\tau^*$ has at least
$k_g$ edges from $g$ (it is a feasible \cst solution), we have $|\tau^*(E \setminus A)\bigcap (g\setminus A)|\ge  k_g - |A|$.
\begin{itemize}
\item Suppose that there exists $\overline{j} \in \tau^*(E \setminus A) \cap B$ such that $\overline{j} \notin L$. Then
since $\overline{j} \in B \setminus L$, it follows that $( k_g - |A|) \sum_{j \in B \setminus L} x_j \geq k_g - |A|$,
hence the constraint is satisfied.
\item The remaining case has $\tau^*(E \setminus A) \cap B \sse L$.
In words, $B$ cuts $g \setminus A$ from $r$ using only leaf edges; so $B\supseteq \tau^*(E \setminus A) \bigcap
(g\setminus A)$. Thus $\sum_{j \in B \cap (L \setminus A)} x_j \geq |\tau^*(E \setminus A)\bigcap (g\setminus A)| \ge
k_g - |A|$.
\end{itemize}
In both the above cases, constraint~\eqref{LP-cst-2} is satisfied.

\newcommand{\mcexcept}{\textsf{MinCutWithExceptions}}
\newcommand{\minl}{\textsf{MinCUT}}
\newcommand{\lhs}{\textsf{CUT}}
\newcommand{\supere}{\texttt{e}}

\paragraph{Solving ~\ref{lp:cst}.} Since $\mathsf{LP}_\mathsf{CST}$ has exponentially many constraints, in order to solve it in polynomial time, we need a
separation oracle. Again we focus on constraints (\ref{LP-cst-2}), since other constraints are only polynomially
many. We observe that this separation oracle reduces to the following problem.
\smallskip

\noindent \textbf{Problem} \mcexcept: Given as input a tree $T$ rooted at $r$ with leaves $L$
and cost $\ell(e)$ on each edge $e$ and an integer $D \geq 0$, the goal is to find a minimum cost cut that separates $r$ from any $D$ leaves. 

To see how this suffices to separate constraints~(\ref{LP-cst-2}), consider any fixed  $g \in \cG$ and all $A\sse g$ with $|A| = \eta$ (finally we iterate over all $g\in
\cG$ and $0\le \eta\le n$). Then $k_g - |A|$ (the right-hand-side of the constraints) is also fixed. Given $x_j$
values, we would like to find $A \subseteq g$ with $|A| = \eta$ and $B \in \cut(r, g\setminus A)$ that minimizes the
left-hand-side, and test if this is smaller than $k_g - \eta$.
Formally, we can recast this into $\mcexcept$ as follows: Remove all edges from $E$ that are not on any path from the
root $r$ to a vertex in $g$, and let $T'$ be the resulting tree and this is the input tree to the problem. Note that
leaves of $T'$ are precisely $g$. For all leaf-edges $j \in g$, let $\ell(j) := x_j$; and for all non-leaf $e \in T'
\setminus g$, $\ell(e) := (k_g - \eta) \cdot x_e$. Also set bound $D:= |g| - \eta$.

We next show that $\mcexcept$ can be solved via a dynamic programming.

\begin{lemma}
    The problem $\mcexcept$ can be solved in polynomial time.
\end{lemma}
\begin{proof}
To formally describe our dynamic program, we make some simplifying assumptions. By introducing dummy edges of
infinite cost, we assume without loss of generality, that the tree $T$ is binary and the root $r$ is incident to exactly one edge $e_r$. Hence
every non-leaf edge $e$ has exactly two child-edges $e_1$ and $e_2$. For any edge $e\in T$, let $T_e$ denote the subtree
of $T$ rooted at $e$, i.e. $T_e$ contains edge $e$ and all its descendants.

We define a recurrence for $C[e,k]$ which denotes the minimum cost cut that separates the root of $T_e$ from exactly
$k$ leaves in $T_e$. Note that $C[e_r, \, D]$ gives the optimal value.

For any leaf-edge $f$ set:
$$
C[f,k]= \left\{
\begin{array}{ll}
0& \mbox{ if } k=0,\\
\ell(f) & \mbox{ if } k=1, \mbox{ and }\\
\infty & \mbox{ otherwise.}
\end{array}\right.
$$
For any non-leaf edge $e$ with children $e_1$ and $e_2$, set:

$$ 
C[e,k]= \left\{
\begin{array}{ll}
0& \mbox{ if } k=0;\\
\displaystyle \min_{k_1+k_2=k}\{ C[e_1,k_1] + C[e_2,k_2]\}& \mbox{ if } 1\le k < |L\cap T_e|; \\
\min\begin{cases} \ell(e), \\ \displaystyle \min_{k_1+k_2=k}\{ C[e_1,k_1] + C[e_2,k_2] \} \end{cases}& \mbox{ if } k = |L \cap T_e| ;\\
\infty & \mbox{ otherwise.}
\end{array}\right.
$$
It can be checked directly that this recurrence computes the desired values  in polynomial time.
\end{proof}

\subsection{LP Relaxation for \lcst}\label{subsec:lcst-LP}
We formulate the following linear relaxation for tree instances of latency covering Steiner tree.
\begin{align}
\tag{$\mathsf{LP}_\mathsf{LCST}$} \label{lp:lcst}
  \min  \quad &  \frac12\cdot \sum_{\ell \geq 0 } 2^\ell \sum_{g \in \cG}  (1 -y^\ell_g) & \\
  \mbox{s.t.} \quad&   x^\ell_{pe(e)}  \geq  x^\ell_e  & \forall \ell \geq 0, e \in E & \label{LP-lcst-3}  \\
 &     \sum_{j \in E} w_e x^{\ell}_e  \leq  2^\ell & \forall \ell \geq 0 & \label{LP-lcst-4} \\
 & ( k_g - |A|) \sum_{j \in B  \setminus L} x^\ell_j + \sum_{j \in B \cap L \setminus A} x^\ell_j  \geq  (k_g - |A|) \cdot y^{\ell}_g & \forall \ell \geq 0, g \in \cG,  A \subseteq g, B \in \cut(r,g \setminus A)  &\label{LP-lcst-2} \\
 & y^{\ell+1}_g  \geq  y^\ell_g  & \forall \ell  \geq 0,  g \in \cG &\label{LP-lcst-5} \\
 & x^\ell_e  \in  [0,1] & \forall  \ell \geq 0, e \in E &\nonumber \\
 & y^\ell_g \in  [0,1] & \forall  \ell \geq 0, g\in \cG &\nonumber
\end{align}

To see that this is a valid relaxation, let $\opt$ denote the optimal path. For any $\ell\ge 0$ let $\opt(2^\ell)$
denote the prefix of length $2^\ell$ in $\opt$. We construct a feasible integral solution to~\ref{lp:lcst} as follows. The
variable $x^\ell_e$ indicates if edge $e$ lies in $\opt(2^\ell)$. The indicator variable $y_g^\ell$ has value one
if and only if group $g$ is covered by $\opt(2^\ell)$, i.e. at least $k_g$ vertices of $g$ are contained in $\opt(2^\ell)$. Constraints (\ref{LP-lcst-3}) follow from the fact that
$\opt(2^\ell)$ is a path starting at $r$. Constraints (\ref{LP-lcst-4}) say that the edges in $\opt(2^\ell)$ have a
total weight of at most $2^\ell$, which is clearly true. Note that for each $\ell \geq 0$, there is a set of
constraints (\ref{LP-lcst-2}) that is similar to the constraints (\ref{LP-cst-2}) in $\lpcst$; the validity of these constraints~\eqref{LP-lcst-2} can be shown exactly as for~\eqref{LP-cst-2}. Constraints
(\ref{LP-lcst-5}) enforce the fact that if group $g$ is covered by $\opt(2^\ell)$ then it must be covered by
$\opt(2^{\ell+1})$ as well, which is trivially true. Now consider the objective value: the total contribution of a group $g$ that is covered by
\opt at some time $t\in (2^k,\,2^{k+1}]$ is $\frac12\cdot \sum_{\ell=0}^k 2^\ell\le 2^k$. Thus the objective value of this
integral solution is at most \opt.

We can ensure by standard scaling arguments, at the loss of a $1+o(1)$ factor in the objective, that all distances are polynomially bounded. This implies that the length of any optimal path is also polynomial, and so it suffices to consider $O(\log n)$ many values of $\ell$. Thus the number of variables in~\ref{lp:lcst} is polynomial.
Note that constraints (\ref{LP-lcst-2}) are exponentially many. However, for each fixed $\ell$ and $g$, we can use 
the same separation oracle that we used for the constraints (\ref{LP-cst-2}) of $\lpcst$.

\subsection{Rounding Algorithm for \lcst}\label{subsec:lcst-algo}
Before presenting our algorithm for \lcst, we discuss the basic rounding scheme from~\cite{KonjevodRS02} (which is an extension of~\cite{GKR00}) and some of its useful properties.

\newcommand{\jkrs}{J_{\textsf{krs}}}

\begin{algorithm}[h!] \caption{\algkrs ~\cite{KonjevodRS02}} \label{alg:krs} 
\textbf{INPUT}: Undirected tree $T = (V, E)$ rooted at $r$; $z_e \in [0,1]$, such that for all $e \in
E$, $z_{pe(e)} \geq z_e$.
    \begin{algorithmic}[1]
    \State $S \leftarrow \emptyset$.
    \State For each $e \in E$ incident to the root $r$, add $e$ to $S$ with probability $z_{e}$.
    \State For each $e \in E$ such that $pe(e) \in S$, add $e$ to $S$ with probability $\frac{z_{e}}{z_{pe(e)}}$.
    \end{algorithmic}
\textbf{OUTPUT}: The connected component (tree) $S$.

\end{algorithm}

\begin{proposition}[\cite{KonjevodRS02}]
    \label{prop:krs-edgesel}
    Each edge $e$ is included in the final solution of  $\algkrs$ with probability $z_e$.
\end{proposition}
\begin{proof}
We prove this by induction on the depth of edge $e$ from $r$. The base case involves edges incident to the root $r$, where this property is clearly true. For the inductive step, assume that  the parent edge $pe(e)$ of $e$ is included with probability $z_{pe(e)}$; then by the algorithm description, edge $e$ is included with probability $z_{pe(e)} \cdot \frac{z_{e}}{z_{pe(e)}} = z_e$. 
\end{proof}

\begin{definition}[KRS properties]
Consider any $z\in[0,1]^E$, $g \in \cG$, $R(g) \subseteq g$ and $0\le r_g \le |R(g)|$. We say that $(z, R(g), r_g)$
satisfies the KRS properties if it satisfies the following:
\begin{alignat}{10}
z_{pe(e)} &\,\,\geq \,\,z_e  & \qquad \forall e \in E  \label{krs-2}\\
\sum_{j \in T(e) \cap R(g)} z_{j} &\,\,\leq \,\, r_g \cdot z_e  & \qquad \forall e \in E  \label{krs-3}
\end{alignat}
\noindent where $T(e)$ is the subtree below (and including) edge $e$.
\end{definition}

The first property  (\ref{krs-2}) is the same as the constraints (\ref{LP-lcst-3}). The second property (\ref{krs-3}) is
a Lipschitz-type condition which implies that  conditional on any edge $e$ being chosen, its subtree $T(e)$ can contribute at
most $r_g$ to the requirement of $R(g)$.  

\begin{lemma}[\cite{KonjevodRS02}]
    \label{lem:krs-main}
Suppose that $(z, R(g), r_g)$ satisfies the KRS properties. Let $L_{krs}$ denote the set of leaves that are covered
by $\algkrs$ with input  $\{z_e : e \in E\}$. Consider any constant $\delta \in [0,1]$. Then for any $g \in \cG$, 
$$\Pr \Big[ |L_{krs} \cap R(g)| \leq (1 - \delta) \mu_g \Big]
\leq \exp \Big(-  \frac{\delta^2 \cdot \mu_g}{2 + r_g ( 1+ \ln |R(g)|)}\Big)$$

\noindent where $\mu_g := \Ex \big[| L_{krs} \cap R(g)|\big] = \sum_{j \in R(g)} z_j$.
\end{lemma}
\begin{proof}
\newcommand{\lca}{\texttt{lca}}
We only give a sketch of the proof, since this is implicit in \cite{KonjevodRS02}. For any $j, j' \in R(g)$, we
say that  $j \sim j'$ if and only if (1) $j \neq j'$ and (2) the least common ancestor $\lca(j, j')$ of $j$ and $j'$ is
not $r$.   Define
    $$\Delta_g := \sum_{j, j' \in R(g): j \sim j', z_{\lca(j, j')} >0 } \frac{ z_j \cdot z_{j'}}{z_{\lca(j,j')}}$$

    In Theorem~3.2 in \cite{KonjevodRS02}, Konjevod \etal showed using the KRS properties that
    $$\Delta_g \leq \mu_g ( r_g - 1 + r_g \ln |R(g)|)$$

We note that the proof of Theorem~3.2 implies this, although it is stated only for $\mu_g = r_g$.
Further, they used this bound in Jansen's inequality to obtain for any $\delta \in [0,1]$,
    $$\Pr \Big[ |L_{krs} \cap R(g)| \leq (1 - \delta) \mu_g \Big] \leq \exp\Big( - \frac{\delta^2 \mu_g }{2 + \Delta_g /\mu_g}\Big)$$
Above $\mu_g =\Ex \big[| L_{krs} \cap R(g)|\big] = \sum_{j\in R(g)} \Pr[j\in L_{krs}] = \sum_{j\in R(g)} z_j$, by
Proposition~\ref{prop:krs-edgesel}. Combining the above two inequalities yields the lemma.
\end{proof}


\def\xo{\overline{x}}
\def\yo{\overline{y}}
\def\eo{\overline{E}}
\def\xm{\tilde{x}}

\medskip
We are now ready to present our algorithm to round \ref{lp:lcst}, described formally as $\alglcst$ below. Let $(\xo, \yo)$ denote a fixed optimal solution to $\lplcst$. The algorithm proceeds in {\em phases} $\ell=0,1,2,\cdots$ where the $\ell^{th}$ phase rounding uses variables with superscript $\ell$ in \ref{lp:lcst}.
Let $\eo^\ell := \{ e \in E \; | \; \xo^\ell_e \geq 1/4
\}$. Observe that $\eo^\ell$ forms a tree rooted at $r$ due to the constraints (\ref{LP-lcst-3}). The edges in
$\eo^\ell$ are added to  our solution with probability one by the $\ell$-th phase of our algorithm. Tree $\tilde{T}^\ell$ is obtained from $T$ by contracting edges $\eo^\ell$. Let $R^\ell(g):= g
\setminus \eo^\ell$ and $r^\ell_g = k_g - |g \cap \eo^\ell|$ denote the residual vertices of group $g$ and its residual requirement, in phase $\ell$. In the subsequent analysis, we will show that the algorithm satisfies a group with constant probability in every phase when it is substantially covered (say to
extent $\frac12$) by the fractional solution $(\xo, \yo)$.

\begin{algorithm}[h!] \caption{$\alglcst$}  \label{alg:lcst} 
\textbf{INPUT}: Tree $T$ with edge lengths, root $r$, groups $\cG$ and requirements $\{k_g\}_{g\in \cG}$.
\begin{algorithmic}[1]
  \State  $\pi \leftarrow \emptyset$.
  \State Let $(\xo,\yo)$ be an optimal solution to $\lplcst$.
  \For {$\ell = 0, 1, 2, ...$}
		\State $\eo^\ell \leftarrow \{ e \in E \; | \; \xo^\ell_e \geq 1/4 \}$, $R^\ell(g) \leftarrow g \setminus \eo^\ell$ and $r^\ell_g \leftarrow k_g - |g \cap \eo^\ell|$.
		\State Shrink all edges in $\eo^\ell$ in $T$ and let $\tilde{T}^\ell$ be the resulting tree with the edge set $\tilde{E}^\ell:= E \setminus \eo^\ell$.
		\State \label{alg:lcst-preproc} Obtain solution $\xm^\ell$ from $\xo^\ell$ using Lemma~\ref{lem:goodx-2}.
		\State For each $e \in \tilde{E}^\ell$, $z^\ell_e \leftarrow 4 \xm^\ell_e$; note that $z^\ell_e \in[0,1]$.
    \State  $ \cS^\ell \leftarrow \emptyset$.
    \State  \textbf{repeat} the following $6 ( 3 + \log g_{max})$ times:
    \State \label{alg:lcst-krs} \hspace{\algorithmicindent}  $\tau^\ell \leftarrow\;$ the tree produced by $\algkrs$ with $z = z^\ell$ on tree $\tilde{T}^\ell$
    \State \hspace{\algorithmicindent}    Add $\tau^\ell$ to $\cS^\ell$
    \State Combine all trees in $\cS^\ell$ with  $\eo^\ell$ and take an Euler tour $P^\ell$ of the resulting tree.
    \State  \textbf{if} $P^\ell$ has weight at most $ 192 ( 3 + \log g_{max}) \cdot 2^\ell$ \textbf{then}
    \State  \hspace{\algorithmicindent} $\pi \leftarrow \pi \cdot P^\ell$.
  \EndFor
\end{algorithmic}
\textbf{OUTPUT}: Path $\pi$ originating from $r$.
\end{algorithm}

In each phase $\ell\ge 0$ we preprocess (in Line~\ref{alg:lcst-preproc}) $\xo^\ell$ to obtain $\xm^\ell$ as described in the next lemma.

\begin{lemma}
    \label{lem:goodx-2}
For any $\ell \geq 0$, we can find in polynomial time  $\xm^\ell_e \in [0, \xo^\ell_e]$, $\forall e \in E\setminus
\eo^\ell$ such that $\forall g\in \cG$:
\begin{enumerate}
    \item $(\xm^\ell,R^\ell(g),r^\ell_g)$ satisfies the KRS-properties in tree $\tilde{T}^\ell$.
    \item $\sum_{j \in R^\ell(g)} \xm^\ell_j \,\, \geq \,\, r^\ell_g \cdot \yo^\ell_g$ \,\,(coverage property).
\end{enumerate}
\end{lemma}
\begin{proof}
Fix any $\ell\ge 0$. To reduce notation, we drop the superscript $\ell$ from $\tilde{T}$, $\eo$, $\xo$, $\yo$, $\xm$, $R(g)$ and $r_g$ throughout this proof. Consider constraints \eqref{LP-lcst-2} of $\lplcst$. Fix a group $g \in \cG$ and let $A:=g\cap \eo$. Consider tree $\tilde T$
as a flow network with each leaf edge $f$ having capacity $\xo_f$ and each non-leaf edge $e$  having capacity $r_g\cdot \xo_e$. The root $r$ is the source and leaves $R(g)=g\setminus A$ are the sinks. Then constraints
(\ref{LP-lcst-2}) imply that the min cut separating $r$ from $R(g)$ has value at least $r_g\cdot \yo_g$: note that although these
constraints are for the original tree $T$, they imply similar constraints for $\tilde T$ since $\tilde T$ is obtained from $T$ by
edge-contraction.\footnote{In particular every cut $B'$ separating $r$ from $g\setminus A$ in $\tilde T$ is also a cut
separating $r$ from $g\setminus A$ in $T$.}
Hence there must exist a max-flow of volume at
least $r_g\cdot \yo_g$ from $r$ to $R(g)$ in the above network. Let
$\xm_f$ denote the volume of this flow into each leaf edge $f\in R(g)$; clearly we have that $\xm_f \le \xo_f$ (due to
capacity on leaves) and:
\begin{equation}\label{eq:krs-cond1}
\sum_{j\in R(g)} \xm_j \ge r_g\cdot \yo_g.
\end{equation}

Moreover, by the capacities on non-leaves,
\begin{equation}\label{eq:krs-cond2}
\sum_{j \in T(e) \cap R(g)} \xm_{j} \leq  r_g \cdot \xo_e, \; \quad \forall e \in E\setminus \eo
\end{equation}

We can use the above procedure on each group $g\in \cG$ separately, to compute $\xm_f$ for all leaf edges $f\in E\setminus \eo$; this is
well-defined since groups are disjoint. For each non-leaf edge $e\in E\setminus \eo$ set $\xm_e:=\xo_e$. Thus we have  $0\le \xm_e\le
\xo_e$ for all $e\in E\setminus \eo$. Observe that this computation can easily be done in polynomial time.

Now,~\eqref{eq:krs-cond2} implies the second KRS property~\eqref{krs-3}. Property (\ref{krs-2})  follows, since for
each $e \in E\setminus \eo$, we have $\xm_{pe(e)} = \xo_{pe(e)} \geq \xo_e \geq \xm_e$; the first inequality is due to constraint
(\ref{LP-lcst-3}) of $\lplcst$. Finally,~\eqref{eq:krs-cond1} implies the coverage property claimed in the lemma.
\end{proof}

\subsection{Analysis}\label{subsec:lcst-analysis}
For any group $g$, define $\ell(g)$ to be the smallest  $\ell \geq 0 $ such that $\yo^\ell_g \geq 1/2$. Then it follows
that for any $\ell \geq \ell(g)$, $\yo^{\ell}_g \geq 1/2$ due to constraints (\ref{LP-lcst-5}) of $\lplcst$.  In words,
the optimal fractional solution covers group $g$ to an extent of at least half within time $2^{\ell(g)}$. Consider any group $g \in
\cG$, $\ell \geq \ell(g)$ and  a tree $\tau^\ell$ in Line~\eqref{alg:lcst-krs} of $\alglcst$. Since all edges in $\eo^\ell$ are
included in $P^\ell$ with probability 1, group $g$ is covered by $P^\ell$ if and only if at least $r^\ell_g$ vertices
in its residual set $R^\ell(g)$ are covered by $\tau^\ell$. This motivates us to derive the following.

\begin{lemma}
    \label{lem:lcst-prob-fail}
    For any $g \in \cG$ and $\ell \geq \ell(g)$, $$\Pr [ | \tau^\ell \cap R(g) | < r_g ] \leq \exp \left( - \frac{1}{2 ( 3+ \ln g_{max})}\right).$$
\end{lemma}
\begin{proof}
From Lemma~\ref{lem:goodx-2} it follows that $(\xm^\ell, R^\ell(g), r^\ell_g)$ satisfies the KRS properties on tree
$\tilde{T}^\ell$. Since $z^\ell=4\cdot \xm^\ell$, $(z^\ell, R^\ell(g), r^\ell_g)$  also satisfies the KRS properties. Furthermore, using $\yo^\ell_g \geq \frac{1}{2}$ and the coverage property in
Lemma~\ref{lem:goodx-2},
$$\mu_g^\ell \,\, := \,\,  \Ex [ |\tau^\ell \cap R^\ell(g)| ]  \,\, =  \,\, \sum_{j \in R^\ell(g)} z^\ell_j   =  \,\, 4 \cdot \sum_{j \in R^\ell(g)} \xm^\ell_j  \,\, \geq  \,\, 4  \cdot r^\ell_g \cdot \yo^\ell_g    \,\, \geq  \,\, 2 r^\ell_g. $$
Here we also used Proposition~\ref{prop:krs-edgesel} that $\Pr[j\in \tau^\ell]=z^\ell_j$.  By applying
Lemma~\ref{lem:krs-main} with $\delta = 1/2$, we have
$$\Pr \Big[ |\tau^\ell \cap R(g)| < r_g \Big] \quad \leq \quad \exp \Big(-  \frac{ r^\ell_g}{2(2 + r^\ell_g ( 1+ \ln |R^\ell(g)|))}\Big) \quad \leq \quad \exp \Big( - \frac{1}{2 ( 3+ \ln g_{max})}\Big).$$
This proves Lemma~\ref{lem:lcst-prob-fail}.\end{proof}

\begin{lemma}
    \label{lem:lcst-prob}
Consider any group $g \in \cG$ and $\ell \geq \ell(g)$. The probability that $P^\ell$ has a total weight of at most $
192 ( 3 + \log g_{max}) \cdot 2^\ell$ and  covers $g$ is at least $3/4$.
\end{lemma}
\begin{proof}
By Proposition~\ref{prop:krs-edgesel}, we know that each edge $e \in \tilde{E}^\ell$ is included in $\tau^\ell$ with
probability $z^\ell_e = 4 \xm_e^\ell\le 4\xo_e^\ell$.    Since for all $e \in \eo^\ell$, $\xo_e \geq 1/4$, the expected total weight of
the edges in $\eo^\ell$ and $\tau^\ell$ is upper bounded by
$$\sum_{e \in \eo^\ell} w_e + \sum_{e \in \tilde{E}^\ell} w_e \cdot 4 \xm^\ell \quad \leq \quad 4 \sum_{e \in E} w_e\cdot \xo^\ell_e \quad \leq \quad 4 \cdot 2^\ell$$
The last inequality is due to the constraints (\ref{LP-lcst-4}). Hence the expected cost of $P^\ell$ is at most  $ 24 (
3 + \log g_{max}) \cdot 2^\ell.$ Markov's inequality immediately gives that the total weight of $P^\ell$ is greater than
$ 192 ( 3 + \log g_{max}) \cdot 2^\ell$ with probability at most $1/8$. Since $\tau_\ell$ is sampled $ 6 ( 3 + \log g_{max})$ times independently, from Lemma~\ref{lem:lcst-prob-fail}, we know that group $g$ is not covered by $P^\ell$
with probability at most $1 / e^3 \leq 1/8$.  Hence the lemma follows.
\end{proof}

Fix any group $g\in\cG$, and $\ell\ge \ell(g)$. Among $P^{\ell(g)}$, $P^{\ell(g)+1}$, ... , $P^{\ell}$, consider the
paths that are added to $\pi$. Clearly the total weight of such paths is at most $O( \log g_{max} \cdot 2^\ell)$. By
Lemma~\ref{lem:lcst-prob}, the probability that none of these paths covers $g$ is at most $\frac{1}{4^{\ell -
\ell(g)+1}}$. Hence the expected cover time of $g$ is at most $$\sum_{\ell \geq \ell(g)} O(\log g_{max}) \cdot 2^\ell
\cdot \frac{1}{4^{\ell - \ell(g)+1}} = O( \log g_{max} \cdot 2^{\ell(g)}).$$ Thus the expected total cover time is at most  $O( \log g_{max}) \cdot \sum_{g \in \cG} 2^{\ell(g)}$.

By definition of $2^{\ell(g)}$ being the ``half completion time'' in the LP, we know
$$\opt 
 \geq \,\, \frac12\cdot \sum_{\ell \geq 0 } 2^\ell \,\,\sum_{g \in \cG}  (1 -\yo^\ell_g) 
 \geq \frac{1}{2} \sum_{g \in \cG} 2^{\ell(g) -1} \left(1 -  \yo^{\ell(g)-1}_g\right)  
 \geq  \frac{1}{8} \cdot \sum_{g \in \cG} 2^{\ell(g)}. $$

Thus we obtain that $\alglcst$ is an $O(\log g_{max})$-approximation for \lcst on tree instances. Using probabilistic tree embedding~\cite{FakcharoenpholRT04}, we conclude that $\alglcst$ yields an $O(\log g_{max} \cdot \log |V|)$-approximation for general metrics, thereby proving  Theorem~\ref{thm:lcst}.

\section{Weighted Stochastic Submodular Ranking} \label{sec:stochastic}
In this section, we study the Weighted Stochastic Submodular Ranking problem (\wssr). The input consists of a set
$\cA=\{X_1,...,X_n\}$ of $n$ independent random variables (stochastic elements), each over domain $\Delta$, with integer
lengths $\{\ell_j\}_{j=1}^n$ (deterministic), and $m$ monotone submodular functions $f_1,...,f_m: 2^{\Delta} \rightarrow
[0,1]$ on groundset $\Delta$. We are also given the distribution (over $\Delta$) of each stochastic element
$\{X_j\}_{j=1}^n$. (We assume explicit probability distributions, i.e. for each $X_j$ and $b\in \Delta$ we are given $\Pr[X_j=b]$.)
The realization $x_{j}\in \Delta$ of the random variable $X_{j}$ is known immediately after
scheduling it. Here, $X_j$ requires $\ell_j$ units of time to be scheduled; if $X_j$ is started at time $t$ then it completes at time $t+\ell_j$ at which point its realization $x_j\in \Delta$ is also known.
A feasible solution/policy is an
adaptive ordering of $\cA$, represented naturally by a decision tree with branches corresponding to the realization of
the stochastic elements. We use $\langle \pi(1),\ldots,\pi(n)\rangle$ to denote this ordering, where each $\pi(l)$ is a
random variable denoting the index of the $l$-th scheduled element.

%

The cover time  $\cov(f_i)$ of any function $f_i$ is defined as the earliest time $t$ such that $f_i$ has value one for the 
realization of the elements completely scheduled within time $t$. More formally, $\cov(f_i)$ is the earliest time $t$
such that $f_i(\{x_{\pi(1)},...,x_{\pi(k_t)}\})$ is equal to $1$ where $k_t$ is the maximum index such that $\ell_{\pi(1)} +
\ell_{\pi(2)} + ... + \ell_{\pi(k_t)} \leq t$. If the function value never reaches one (due to the stochastic nature of
elements) then $\cov(f_i)= \ell_1 + \ell_2 + ... + \ell_n$ which is the maximum time of any schedule. Note that the cover time is a random value. The goal is to
find a policy that (approximately) minimizes the expected total cover time $\Ex\left[\sum_{i\in [m]} \cov(f_i)
\right].$

\subsection{Applications}

Our stochastic extension of submodular ranking captures many interesting applications.

\paragraph{Stochastic Set Cover.} We are given as input a ground set $\Delta$, and a collection $\cS \subseteq
2^\Delta$ of (non-stochastic) subsets. There are stochastic elements $\{X_j : j \in [n]\}$, each defined over $\Delta$
and having respective costs $\{\ell_j:j\in[n]\}$. The goal is to give an adaptive policy that hits all sets in $\cS$ at the minimum expected cost. This problem was studied in \cite{GoemansV06, MunagalaSW07, LiuPRY08}. The problem can be shown
as an instance of \wssr with a single monotone submodular function $f_1(A) := \frac{1}{|\cS|} \sum_{S \in \cS}  \min\{1, |A \cap S|\}$ and parameter $\epsilon=1/|\cS|$.

\paragraph{Shared Filter Evaluation.}  This problem was introduced by \cite{MunagalaSW07}, and the result was
improved to an essentially optimal solution in \cite{LiuPRY08}. In this problem, there is a collection of independent
``filters'' $X_1, X_2, ...., X_n$, each of which gets evaluated either to $\true$ or $\false$. For each filter $j\in[n]$, we are
given the ``selectivity'' $p_j=\Pr[X_i \mbox{ is true}]$ and the cost $\ell_j$ of running the filter. We are also given
a collection $\cQ$ of queries, where each query $Q_i$ is a conjunction of a subset of queries. We would like to
determine each query in $\cQ$ to be $\true$ or $\false$ by (adaptively) testing filters of the minimum expected cost. In order
to cast this problem as \wssr, we use $\Delta = \bigcup_{j=1}^n \{Y_j,\, N_j\}$; for each $j\in[n]$, $X_j=Y_j$
with probability $p_j$, and $X_j=N_j$ with the remaining probability $1-p_j$.
We create one monotone submodular function:
$$f_1( A) := \frac{\displaystyle \sum_{Q_i \in \cQ} \min \left\{1, \,\,
|A\cap \{N_j: j\in Q_i\}| \, + \,  \frac1{|Q_i|}\cdot |A\cap \{Y_j: j\in Q_i\}| \right\}}{|\cQ|}$$
(Note that a query $Q_i$ gets evaluated to: $\false$ if any one of its filters is $\false$, and $\true$ if all its filters are $\true$.) Here the parameter $\epsilon=1/\left(|\cQ|\max_i |Q_i|\right)$.

We note that
the Shared Filter Evaluation problem can be studied for a latency type of objective also. In this case, for each
query $Q_i \in \cQ$, we create a separate submodular function:
$$f_i ( A) := \min \left\{1, \,\,
|A\cap \{N_j: j\in Q_i\}| \, + \,  \frac1{|Q_i|}\cdot |A\cap \{Y_j: j\in Q_i\}| \right\}$$ In this case, the \wssr problem corresponds precisely to filter evaluation that minimizes the {\em average time} to answer queries in
$\cQ$. The parameter $\epsilon=1/\left(\max_i |Q_i|\right)$.

\paragraph{Stochastic Generalized Min Sum Set Cover.}  We are given as input a ground set $\Delta$, and a
collection $\cS \subseteq 2^\Delta$ of (non-stochastic) subsets with  requirement $k(S)$ for each $S\in\cS$. There are
stochastic elements $\{X_j : j\in [n]\}$, each defined over $\Delta$. Set $S\in \cS$ is said to be completed when at
least $k(S)$ elements from $S$ have been scheduled. The goal is to find an adaptive ordering of $[n]$ so as to minimize
the expected total completion time. This can be reduced to \wssr by defining function $f^S(A) := \min\{1,\,
|A\cap S|/k(S)\}$ for each $S\in \cS$; here $\epsilon=1/k_{max}$ where $k_{max}$ denotes the maximum requirement.

For this problem, our result implies an $O(\log k_{max})$-approximation
to adaptive policies. However, for non-adaptive policies (where the ordering of elements is fixed {\em a priori}), one can obtain a better $O(1)$-approximation algorithm by
combining the Sample Average Approximation (SAA) method \cite{KleywegtSH02, CharikarCP05} with $O(1)$-approximations
known for the non-stochastic version \cite{BansalGK10, SkutellaW11}.

\medskip
We also note that the analysis in~\cite{AzarG11} for the deterministic submodular ranking was only for elements having unit sizes. Our analysis also holds under non-uniform sizes.

\subsection{Algorithm and Analysis}
    \label{sec:sto-ag}

We consider adaptive policies: this chooses at each time $\ell_{\pi(1)} + \ell_{\pi(2)} + ... + \ell_{\pi(k-1)}$, the
element $$X_{\pi(k)} \in \cA \setminus \{X_{\pi(1)}, X_{\pi(2)}, X_{\pi(3)}, ..., X_{\pi(k-1)}\}$$ after observing the
realizations $x_{\pi(1)},...,x_{\pi(k-1)}$. So it can be described as a decision tree.
Our main result is an $O(\log \frac{1}{\eps})$-approximate adaptive policy, which proves Theorem~\ref{thm:sto-sr}. This result is again inspired by our
simpler analysis of the algorithm from~\cite{AzarG11}.

To formally describe our algorithm, we quickly define the probability spaces we are concerned with.  We use $\Omega
= \Delta^n$ to denote the outcome space of $\cA$.
We use the same notation $\Omega$ to denote the probability space induced by this outcome space. For any $S \subseteq
\cA$ and its realization $s$, let $\Omega(s)$ denote the outcome subspace that conforms to $s$. We can naturally define
the probability space defined by $\Omega(s)$ as follows: The probability that $w \in \Omega(s)$ occurs is $\Pr_{\Omega}
[w] / \Pr_{\Omega}[ \Omega(s)]$; we also use $\Omega(s)$ to denote this probability space.

The main algorithm is given below and is a natural extension of the deterministic algorithm~\cite{AzarG11}. Let $\alpha: = 1 + \ln( \frac{1}{\eps})$. In
the output, $\pi(l)$ denotes the $l$th element in $\cA$ that is scheduled.

\begin{algorithm}[ht!] \caption{$\agsto$} \label{alg:ag-sto-main} 
\begin{algorithmic}[1]
     \State \textbf{INPUT}: $\cA = \{X_1,...,X_n\}$ with $\{\ell_1, ..., \ell_n\}$; $f_i : 2^{\Delta} \rightarrow [0,1]$, $i \in [m]$.
     \State $S \leftarrow \emptyset$.  ($S$ are the elements completely scheduled so far, and $s$ their instantiation.)
             \While{there exists function $f_i$ with $f_i(s) <1$}
             \State Choose element $X_e$ as follows, 
             \label{step:greedy-stoch} $$X_e = \arg\max_{X_e \in \cA \setminus S} \,\,   \frac{ \Ex_{\ \Omega(s)} \left[  \sum_{i \in [m], f_i(s) <1} \,\,  \frac{f_i(s \cup \{X_e\}) - f_i (s)}{1 - f_i(s)}  \right]}{\ell_e} $$
        \State $S \leftarrow S \bigcup \{X_e\}$.
        \State $\pi(|S|) \gets X_e$.  Schedule $X_e$ and observe its realization.
           \EndWhile
\State \textbf{OUTPUT}: An adaptive ordering $\pi$ of $\cA$.
           \end{algorithmic}
\end{algorithm}

Observe that taking expectation over $\Omega(s)$ in Step~\ref{step:greedy-stoch} is the same as expectation over
the distribution of $X_e$ since $X_e\not\in S$ and the elements are independent. This value can be computed exactly since we have an explicit probability distribution of $X_e$. Also note that this algorithm
implicitly defines a decision tree. We will show that $\agsto$ is an $O(\ln (\frac{1}{\eps}))$-approximation algorithm for \wssr.

To simplify notation, without loss of generality, we assume that $\alpha$ is
an integer. Let $R(t)$ denote the (random) set of functions that are not satisfied by \agsto before time $t$. Note that
the set $R(t)$ includes the functions that are satisfied exactly at time $t$. Analogously, the set $R^*(t)$ is defined
for the optimal policy. For notational convenience, we use $i \in R(t)$ interchangeably with $f_i \in
R(t)$. Let $C(t) := \{f_1,...,f_m\} \setminus R(t)$ and $C^*(t):= \{f_1,...,f_m\} \setminus R^*(t)$.
Note that all the sets $C(\cdot)$, $C^*(\cdot)$, $R(\cdot)$, $R^*(\cdot)$ are stochastic. We have that $\alg = \sum_{t
\in [n]} |R(t)|$ and $\opt = \sum_{t \in [n]} |R^*(t)|$ and hence $\alg$ and $\opt$ are stochastic quantities. We show that $\Ex[\alg] = O(\alpha)\cdot \Ex[\opt]$ which suffices to prove the desired approximation ratio.

We are interested in the number of unsatisfied functions at times $\{8 \alpha 2^j \, :\, j \in \mathbb{Z}_+\}$ by
$\agsto$ and the number of unsatisfied functions at times $\{2^j\, :\, j \in \mathbb{Z}_+\}$ by the optimal policy.
Let $R_j := R(8 \alpha 2^j)$ and $R^*_j = R^*(2^j)$. It is important to note that $R_j$ and $R^*_j$ are concerned with
different times, and they are stochastic. For notational simplicity, we let $R_{-1} := \emptyset$.

We show the following key lemma. 
Once we
get this lemma, we can complete the proof similar to the proof of Theorem~\ref{thm:ag-main} via Lemma~\ref{lem:ag-main}.

\newcommand{\fron}{\texttt{Fron}}

\begin{lemma}    \label{lem:sto-ag-main}
    For any $j \geq 0$, we have $\Ex[|R_j|] \leq \frac{1}{4}\Ex[|R_{j-1}|] + \Ex[|R^*_j|]$.
\end{lemma}
\begin{proof}
The lemma trivially holds for $j=0$, so we consider  any $j \geq 1$. For any $t\ge 1$, we use $s_{t-1}$ to denote the
set of elements {\em completely} scheduled by $\agsto$ by time $t-1$ along with their instantiations; clearly this is a
random variable. Also, for $t\ge 1$ let $\sigma(t)\in[n]$ denote the (random) index of the element being scheduled
during time slot $(t-1,t]$. Since elements have different sizes, note that $\sigma(t)$ is different from
$\pi(t)$ which is the $t$-th element scheduled by $\agsto$. Observe that $s_{t-1}$ determines $\sigma(t)$ precisely, but
not the instantiation of $X_{\sigma(t)}$.

 Let $E_j^* \subseteq \cA$ be the (stochastic) set of elements that is completely scheduled by the optimal
policy within time $2^j$. For a certain stochastic set (or elements) $S$, we denote its realization under an
outcome $w$ as $S(w)$. For example, $X_i(w)\in \Delta$ is the realization of element $X_i$ for outcome $w$; and
$E^*_j(w)$ is the set of first $2^j$ elements completely scheduled by $\opt$ (under $w$) along with their realizations.

For any time $t$ and corresponding outcome $s_{t-1}$, define a set function:

$$f^{s_{t-1}}(D) \,\, := \,\, \sum_{i \in [m], f_i(s_{t-1}) <1} \,\,  \frac{f_i(s_{t-1} \cup D) - f_i (s_{t-1})}{1 - f_i(s_{t-1})}, \qquad \forall D\sse \Delta. $$
We also use $f^{s_{t-1}}_i(D)$ to denote the term inside the above summation.

The function $f^{s_{t-1}}: 2^\Delta \rightarrow \mathbb{R}_+$ is monotone and submodular since it is a summation of monotone and submodular functions. We also define
\begin{equation}
\label{eq:def-F} F^{s_{t-1}}(X_e) \,\, :=  \,\, \Ex_{\ w \leftarrow \Omega(s_{t-1})}\left[  f^{s_{t-1}}(X_e(w))
\right], \qquad \forall X_e\in \cA.\end{equation}

Observe that this is zero for elements $X_e\in s_{t-1}$.

\begin{proposition}
    \label{prop:sto-1}
Consider any time $t$ and outcome $s_{t-1}$. Note that $s_{t-1}$ determines $\sigma(t)$. Then:
$$\frac{1}{\ell_{\sigma(t)}} \cdot F^{s_{t-1}}(X_{\sigma(t)}) \,\, \geq \,\, \frac{1}{\ell_i}  \cdot F^{s_{t-1}}(X_i), \qquad \forall X_i\in \cA$$
\end{proposition}
\begin{proof}
At some time $t'\le t$ (right after $s_{t-1}$ is observed) \agsto chose to schedule element $X_{\sigma(t)}$ over all
elements $X_i \in \cA  \setminus s_{t-1}$. By the greedy rule we know that the claimed inequality holds for any $X_i
\in \cA \setminus s_{t-1}$. Furthermore, the inequality holds for any element $X_i \in s_{t-1}$, since here
$F^{s_{t-1}}(X_i) = 0$.
\end{proof}

We now define the {\em expected gain} by \agsto in step $t$ as:
\begin{equation}\label{eq:def-gain}
G_t \quad := \quad \Ex_{s_{t-1}}\, \left[ \frac{1}{\ell_{\sigma(t)}} F^{s_{t-1}}(X_{\sigma(t)}) \right].
\end{equation}
And the expected total gain:
\begin{equation}
\Delta_j \quad := \quad \sum_{t=8\alpha2^{j-1}}^{8\alpha2^j} \,G_t \quad 
\end{equation}
We complete the proof of Lemma~\ref{lem:sto-ag-main}
by upper and lower bounding $\Delta_j$.

\paragraph{Upper bound for $\Delta_j$.} Fix any outcome $w \in \Omega$. Below, all variables are {\em conditioned on $w$} and hence they
are all deterministic. (For ease of notation we do not write $w$ in front of the variables).
\begin{eqnarray*}
\Delta_j \,\, &:=& \,\, \sum_{t=8\alpha2^{j-1}}^{8\alpha2^j} \frac{1}{\ell_{\sigma(t)}} \,\, f^{s_{t-1}}(x_{\sigma(t)})
\quad =  \quad \sum_{t=8\alpha2^{j-1}}^{8\alpha2^j} \frac{1}{\ell_{\sigma(t)}} \quad \sum_{i \in [m] : f_i(s_{t-1}) <1} \,\, f^{s_{t-1}}_i(x_{\sigma(t)}) \\
&\le & \,\, \sum_{t=8\alpha2^{j-1}}^{8\alpha2^j} \frac{1}{\ell_{\sigma(t)}} \quad \sum_{i \in R_{j-1}} \,\,
f^{s_{t-1}}_i(x_{\sigma(t)}) \quad 
\le  \quad \sum_{t\ge 1} \frac{1}{\ell_{\sigma(t)}} \quad \sum_{i \in R_{j-1}} \,\,  f^{s_{t-1}}_i(x_{\sigma(t)}) \\
&=& \,\, \sum_{i \in R_{j-1}} \,\,\, \sum_{k = 1}^n \,\, \frac{ f_i(T_k) - f_i (T_{k-1})}{1 - f_i (T_{k-1})}
\end{eqnarray*}
The first inequality uses the fact that any $i\not\in R_{j-1}$ has $f_i$ already covered before time $8\alpha\,
2^{j-1}$, and so it never contributes to $\Delta_j$. In the last expression, $T_k := \{x_{\pi(1)}, ..., x_{\pi(k)}
\}\sse \Delta$, the first $k$ instantiations seen under $w$. The equality uses the fact that for each $\sum_{j=1}^{k-1}
\ell_{\pi(j)} < t \le \sum_{j=1}^{k} \ell_{\pi(j)}$ we have $s_{t-1}=T_{k-1}$ and $\sigma(t)=k$. Finally, by
Claim~\ref{lem:log-ub}, the contribution of each function $f_i \in R_{j-1}$ is at most $\alpha := 1+\ln
\frac{1}{\eps}$. Thus we obtain $\Delta_j(w)\le \alpha |R_{j-1}(w)|$, and taking expectations,
\begin{equation}
    \label{eqn:upper}
\Delta_j  \quad \leq \quad \alpha \Ex [ |R_{j-1}| ]
\end{equation}

\paragraph{Lower bound for $\Delta_j$.} Consider any $8\alpha2^{j-1} \le t\le 8\alpha2^{j}$. We lower bound $G_t$.
Condition on $s_{t-1}$; this determines $\sigma(t)$ (but not $x_{\sigma(t)}$). Note that $\sum_{i=1}^n \ell_i \cdot
\Pr[X_i\in E^*_j | s_{t-1}] \, \le \, 2^j$ by definition of $E^*_j$ being the elements that are completely scheduled by
time $2^j$ in $\opt$. Hence we have $$\sum_{X_i\in \cA} \, \frac{\ell_i}{2^j}  \, \cdot \, \Pr[X_i\in E^*_j | s_{t-1}]
\, \le \, 1.$$ By applying Proposition~~\ref{prop:sto-1} with the convex multipliers (over $i$) given above,
\begin{align}
\frac{1}{\ell_{\sigma(t)}}  F^{s_{t-1}}(X_{\sigma(t)}) & \geq  \,\, \sum_{X_i \in \mathcal{A}} \frac{\ell_i}{2^j}  \Pr[X_i \in E^*_j | s_{t-1}] \, \cdot \, \frac{1}{\ell_i} F^{s_{t-1}}(X_i) \nonumber\\
& =\,\,  \frac{1}{2^j}  \sum_{X_i \in \mathcal{A}} \Pr[X_i \in E^*_j | s_{t-1}] \,\,  \sum_{x_i \in \Delta} \Pr [ X_i = x_i | s_{t-1} ] \cdot f^{s_{t-1}} (x_i) \nonumber\\
& =\,\,  \frac{1}{2^j} \sum_{X_i \in \mathcal{A}} \,\, \sum_{x_i \in \Delta} \Pr[X_i \in E^*_j  \wedge X_i = x_i | s_{t-1} ] \cdot f^{s_{t-1}} (x_i) \nonumber\\
& =\,\,  \frac{1}{2^j}  \sum_{w \in \Omega(s_{t-1})} \Pr[w | s_{t-1}] \,\, \sum_{X_i \in E^*_j(w)} f^{s_{t-1}}(X_i(w))
\label{eqn:sto-1}
\end{align}

The first equality is by definition of $F^{s_{t-1}}(\cdot)$ from~\eqref{eq:def-F}.  The second equality holds since the optimal policy must decide
whether to schedule $X_i$ (by time $2^j$) without knowing the realization of $X_i$. Now for each $w \in
\Omega(s_{t-1})$, due to submodularity of the function $f^{s_{t-1}}(\cdot)$, we get
{\small \begin{equation}
\sum_{X_i \in E^*_j(w)}  f^{s_{t-1}}(X_i(w))
\,\, \geq  \,\, f^{s_{t-1}}(E^*_j(w)) \,\, =   \,\, \sum_{i \in [m], f_i(s_{t-1}) <1}  \frac{f_i(E^*_j(w)) - f_i (s_{t-1})}{1 - f_i(s_{t-1})} \,\, \geq  \,\, |C_j^*(w)| - |C(t, w)|. \label{eqn:sto-2}
\end{equation}}

Recall that $E^*_j(w)$ denotes the set of elements scheduled by time $2^j$ in \opt~ (conditional on $w$), as well as the
realizations of these elements. The equality comes from the definition of $f^{s_{t-1}}$. The last inequality holds
because $C(t,w) = \{i \in [m] \,:\, f_i(s_{t-1}) = 1\}$ and set $E^*_j(w)$ covers functions $C_j^*(w)$. Combining
(\ref{eqn:sto-1}) and (\ref{eqn:sto-2}) gives:
\begin{equation*}
 \frac{1}{\ell_{\sigma(t)}} F^{s_{t-1}}(X_{\sigma(t)}) \quad \geq  \quad \frac{\Big ( \Ex \left[ |C_j^*| \  | \ s_{t-1}\ \ \right] - \Ex \left[ |C(t)| \ |  \ s_{t-1}\ \ \right] \Big)}{2^j } .
\end{equation*}

By deconditioning the above inequality (taking expectation over $s_{t-1}$) and using~\eqref{eq:def-gain}, we derive:
$$G_t \,\, \geq \,\, \frac{1}{2^j} \cdot \Big (\Ex [ |C^*_j|] - \Ex[ |C(t)|] \Big) \,\, \geq \,\, \frac{1}{2^j} \cdot \Big (\Ex [ |C^*_j|] - \Ex[ |C_j|] \Big),$$
where the last inequality uses $\Ex[ C(t)]$ is non-decreasing and $t\le 8\alpha2^j$.

Now summing over all $t \in [8\alpha 2^{j-1}, 8 \alpha 2^{j})$ yields:
\begin{equation}
    \label{eqn:lower}
\Delta_j \quad   =  \quad \sum_{t =8\alpha 2^{j-1}}^{ 8 \alpha 2^{j}}  G_t  \quad \geq  \quad 4 \alpha \Big (\Ex [ |C^*_j|] -
\Ex[ |C_j|] \Big) \quad = \quad 4 \alpha \Big (\Ex [ |R_j|] - \Ex[ |R^*_j|] \Big).
\end{equation}

Combining (\ref{eqn:lower}) and (\ref{eqn:upper}),  we obtain:
$$4 \alpha ( \Ex[|R_j|] - \Ex[ |R^*_j|]) \quad \leq  \quad \alpha \Ex[ |R_{j-1}|] $$

which simplifies to the desired inequality in Lemma~\ref{lem:sto-ag-main}.\end{proof}

Using exactly the same calculations as in the proof of Theorem~\ref{thm:ag-main} from Lemma~\ref{lem:ag-main}, Lemma~\ref{lem:sto-ag-main} implies an $O(\alpha)$-approximation ratio for \agsto. This completes the proof of Theorem~\ref{thm:sto-sr}.

\section{Conclusion} \label{sec:conclusion}
In this paper we considered the minimum latency submodular cover problem in general metrics, which is a common generalization of many well-studied problems. We also studied the stochastic submodular ranking problem, which generalizes a number of stochastic optimization problems. Both results were based on a new analysis of the algorithm for submodular ranking~\cite{AzarG11}. Our result for stochastic submodular ranking is tight, and any significant improvement (more than a $\log^\delta|V|$ factor) of the result for minimum latency submodular cover would also improve the approximation ratio for Group Steiner Tree, which is a long-standing open problem. An interesting open question is to obtain a poly-logarithmic approximation for stochastic minimum latency submodular cover (on general metrics).
\bibliographystyle{plain}
\bibliography{latencycover_journal}

\begin{thebibliography}{10}

\bibitem{AzarG11}
Y.~Azar and I.~Gamzu.
\newblock Ranking with submodular valuations.
\newblock In {\em 22nd Annual ACM-SIAM Symposium on Discrete Algorithms
  (SODA)}, pages 1070--1079, 2011.

\bibitem{AzarGY09}
Y.~Azar, I.~Gamzu, and X.~Yin.
\newblock Multiple intents re-ranking.
\newblock In {\em 41st Annual ACM Symposium on Theory of Computing (STOC)},
  pages 669--678, 2009.

\bibitem{BansalGK10}
N.~Bansal, A.~Gupta, and R.~Krishnaswamy.
\newblock A constant factor approximation algorithm for generalized min-sum set
  cover.
\newblock In {\em 21st Annual ACM-SIAM Symposium on Discrete Algorithms
  (SODA)}, pages 1539--1545, 2010.

\bibitem{BBHST98}
A.~Bar-Noy, M.~Bellare, M.M. Halld{\'o}rsson, H.~Shachnai, and T.~Tamir.
\newblock On chromatic sums and distributed resource allocation.
\newblock {\em Information and Computation}, 140(2):183--202, 1998.

\bibitem{BCCFV10}
A.~Bhaskara, M.~Charikar, E.~Chlamtac, U.~Feige, and A.~Vijayaraghavan.
\newblock Detecting high log-densities: an $n^{1/4}$ approximation for densest
  $k$-subgraph.
\newblock In {\em 42nd ACM Symposium on Theory of Computing (STOC)}, pages
  201--210, 2010.

\bibitem{CZ05}
G.~Calinescu and A.~Zelikovsky.
\newblock The polymatroid steiner problems.
\newblock {\em Journal of Combinatorial Optimization}, 9(3):281--294, 2005.

\bibitem{CFLP00}
R.D. Carr, L.~Fleischer, V.J. Leung, and C.A. Phillips.
\newblock Strengthening integrality gaps for capacitated network design and
  covering problems.
\newblock In {\em 11th Annual ACM-SIAM Symposium on Discrete Algorithms
  (SODA)}, pages 106--115, 2000.

\bibitem{ChakrabartyS11}
D.~Chakrabarty and C.~Swamy.
\newblock Facility location with client latencies: Linear programming based
  techniques for minimum latency problems.
\newblock In {\em 15th International Conference on Integer Programming and
  Combinatoral Optimization (IPCO)}, pages 92--103, 2011.

\bibitem{CharikarCP05}
M.~Charikar, C.~Chekuri, and M.~P{\'a}l.
\newblock Sampling bounds for stochastic optimization.
\newblock In {\em 9th International Workshop on Randomization and Computation
  (RANDOM)}, pages 257--269, 2005.

\bibitem{CGRT03}
K.~Chaudhuri, B.~Godfrey, S.~Rao, and K.~Talwar.
\newblock Paths, trees, and minimum latency tours.
\newblock In {\em 44th Symposium on Foundations of Computer Science (FOCS)},
  pages 36--45, 2003.

\bibitem{CEK06}
C.~Chekuri, G.~Even, and G.~Kortsarz.
\newblock A greedy approximation algorithm for the group steiner problem.
\newblock {\em Discrete Applied Mathematics}, 154(1):15--34, 2006.

\bibitem{ChekuriP05}
C.~Chekuri and M.~P{\'a}l.
\newblock A recursive greedy algorithm for walks in directed graphs.
\newblock In {\em 46th Annual IEEE Symposium on Foundations of Computer Science
  (FOCS)}, pages 245--253, 2005.

\bibitem{FHR07}
J.~Fakcharoenphol, C.~Harrelson, and S.~Rao.
\newblock The $k$-traveling repairmen problem.
\newblock {\em ACM Transactions on Algorithms}, 3(4), 2007.

\bibitem{FakcharoenpholRT04}
J.~Fakcharoenphol, S.~Rao, and K.~Talwar.
\newblock A tight bound on approximating arbitrary metrics by tree metrics.
\newblock {\em Journal of Computer and System Sciences}, 69(3):485--497, 2004.

\bibitem{FeigeLT04}
U.~Feige, L.~Lov{\'a}sz, and P.~Tetali.
\newblock Approximating min sum set cover.
\newblock {\em Algorithmica}, 40(4):219--234, 2004.

\bibitem{GKR00}
N.~Garg, G.~Konjevod, and R.~Ravi.
\newblock A polylogarithmic approximation algorithm for the group steiner tree
  problem.
\newblock {\em Journal of Algorithms}, 37(1):66--84, 2000.

\bibitem{GoemansV06}
M.X. Goemans and J.~Vondr{\'a}k.
\newblock Stochastic covering and adaptivity.
\newblock In {\em 7th Latin American Symposium on Theoretical Informatics
  (LATIN)}, pages 532--543, 2006.

\bibitem{GolovinK10}
D.~Golovin and A.~Krause.
\newblock Adaptive submodularity: A new approach to active learning and
  stochastic optimization.
\newblock In {\em 23rd Conference on Learning Theory (COLT)}, pages 333--345,
  2010.

\bibitem{GB11}
A.~Guillory and J.A. Bilmes.
\newblock Online submodular set cover, ranking, and repeated active learning.
\newblock In {\em 25th Annual Conference on Neural Information Processing
  Systems (NIPS)}, pages 333--345, 2011.

\bibitem{GuptaNR10a}
A.~Gupta, V.~Nagarajan, and R.~Ravi.
\newblock {Approximation Algorithms for Optimal Decision Trees and Adaptive TSP
  Problems}.
\newblock In {\em 37th International Colloquium on Automata, Languages and
  Programming (ICALP)}, pages 690--701, 2010.

\bibitem{GuptaS06}
A.~Gupta and A.~Srinivasan.
\newblock An improved approximation ratio for the covering steiner problem.
\newblock {\em Theory of Computing}, 2(1):53--64, 2006.

\bibitem{HK03}
E.~Halperin and R.~Krauthgamer.
\newblock Polylogarithmic inapproximability.
\newblock In {\em 35th Annual ACM Symposium on Theory of Computing (STOC)},
  pages 585--594, 2003.

\bibitem{J74}
David~S. Johnson.
\newblock Approximation algorithms for combinatorial problems.
\newblock {\em J. Comput. Syst. Sci.}, 9(3):256--278, 1974.

\bibitem{KleywegtSH02}
A.J. Kleywegt, A.~Shapiro, and T.~Homem de~Mello.
\newblock The sample average approximation method for stochastic discrete
  optimization.
\newblock {\em SIAM Journal on Optimization}, 12(2):479--502, 2002.

\bibitem{KonjevodRS02}
G.~Konjevod, R.~Ravi, and A.~Srinivasan.
\newblock Approximation algorithms for the covering steiner problem.
\newblock {\em Random Structures and Algorithms}, 20(3):465--482, 2002.

\bibitem{LiuPRY08}
Z.~Liu, S.~Parthasarathy, A.~Ranganathan, and H.~Yang.
\newblock Near-optimal algorithms for shared filter evaluation in data stream
  systems.
\newblock In {\em ACM SIGMOD International Conference on Management of Data
  (SIGMOD)}, pages 133--146, 2008.

\bibitem{MunagalaSW07}
K.~Munagala, U.~Srivastava, and J.~Widom.
\newblock Optimization of continuous queries with shared expensive filters.
\newblock In {\em 27th ACM SIGMOD-SIGACT-SIGART Symposium on Principles of
  Database Systems (PODS)}, pages 215--224, 2007.

\bibitem{N09}
V.~Nagarajan.
\newblock {\em Approximation Algorithms for Sequencing Problems}.
\newblock PhD thesis, Tepper School of Business, Carnegie Mellon University,
  2009.

\bibitem{schrijver}
A.~Schrijver.
\newblock {\em Combinatorial optimization: polyhedra and efficiency}.
\newblock Springer-Verlag, Berlin, 2003.

\bibitem{SkutellaW11}
M.~Skutella and D.P. Williamson.
\newblock A note on the generalized min-sum set cover problem.
\newblock {\em Operations Research Letters}, 39(6):433--436, 2011.

\bibitem{Wolsey82}
L.A. Wolsey.
\newblock An analysis of the greedy algorithm for the submodular set covering
  problem.
\newblock {\em Combinatorica}, 2(4):385--393, 1982.

\end{thebibliography}
\end{document}